\documentclass[12pt]{article}
\usepackage{a4wide}
\usepackage{amsthm}

\usepackage[english]{babel}
\usepackage{latexsym}
\usepackage{amssymb}
\usepackage{amsmath}
\usepackage[noadjust]{cite}    
\usepackage{ifthen}
\usepackage{xspace}
\usepackage{paralist}
\usepackage{multirow}
\usepackage{shuffle}
\usepackage{environ}
\usepackage[colorlinks,citecolor=blue]{hyperref}

\usepackage[usenames,dvipsnames]{xcolor}
\usepackage{tikz}
\usetikzlibrary{%
  arrows,%
  positioning,%
  decorations.pathmorphing,%
  decorations.pathreplacing,%
  automata,%
}

\theoremstyle{plain}
\newtheorem{theorem}{Theorem}[]
\newtheorem{proposition}[theorem]{Proposition}
\newtheorem{lemma}[theorem]{Lemma}
\newtheorem{corollary}[theorem]{Corollary}

\theoremstyle{definition}
\newtheorem{definition}[theorem]{Definition}
\newtheorem{example}[theorem]{Example}

\theoremstyle{remark}
\newtheorem{remark}[theorem]{Remark}

\newcommand\pssi{\par\smallskip\indent}
\newcommand\pmsi{\par\medskip\indent}

\newcommand\pssn{\par\smallskip\noindent}
\newcommand\pmsn{\par\medskip\noindent}

\newcommand\pnsi{\par\indent}
\newcommand\pnsn{\par\noindent}

\newcommand{\cA}{\ensuremath{\mathcal{A}}\xspace}
\newcommand{\cB}{\ensuremath{\mathcal{B}}\xspace}

\newcommand{\cH}{\ensuremath{\mathcal{H}}\xspace}

\newcommand{\cP}{\ensuremath{\mathcal{P}}\xspace}

\newcommand{\cS}{\ensuremath{\mathcal{S}}\xspace}
\newcommand{\cT}{\ensuremath{\mathcal{T}}\xspace}

\newcommand{\cW}{\ensuremath{\mathcal{W}}\xspace}

\newcommand{\bA}{\ensuremath{\mathtt{A}}\xspace}
\newcommand{\bC}{\ensuremath{\mathtt{C}}\xspace}
\newcommand{\bG}{\ensuremath{\mathtt{G}}\xspace}
\newcommand{\bT}{\ensuremath{\mathtt{T}}\xspace}

\renewcommand{\phi}{\varphi}

\newcommand{\thetai}{\ensuremath{\theta^{-1}}\xspace}

\let\mathttheta\theta
\renewcommand{\theta}{\texorpdfstring{\ensuremath{\mathttheta}}{theta}\xspace}

\newcommand{\sse}{\subseteq}

\newcommand{\es}{\emptyset}
\newcommand{\sm}{\setminus}

\newcommand{\ie}{i.\,e.,\xspace}
\newcommand{\eg}{e.\,g.,\xspace}

\newcommand{\resp}{resp.,\ }

\newcommand{\am}{\mbox{(anti-)}\allowbreak{}morphic\xspace}

\newcommand{\sett}[2]
	{\left\{#1\mathrel{\left|\vphantom{#1}\vphantom{#2}\right.}#2\right\}}
\newcommand{\set}[1]{\left\{\mathinner{#1}\right\}}

\newcommand{\abs}[1]{\left|\mathinner{#1}\right|}

\newcommand{\ceil}[1]{\left\lceil\mathinner{#1} \right\rceil}

\newcommand{\sdel}{\rightsquigarrow}
\newcommand{\sins}{\shuffle}
\newcommand{\uar}{\uparrow}
\newcommand{\dar}{\downarrow}

\newcommand{\N}{\mathbb{N}}
\newcommand{\Z}{\mathbb{Z}}

\newcommand{\Oh}{\mathcal{O}}

\newcommand{\xra}[1]{\mathbin{\raisebox{-1pt}{$\xrightarrow{#1}$}}}
\newcommand{\xras}[1]{\mathbin{\raisebox{-1pt}{$\xrightarrow{#1}\!\!{}^*\,$}}}

\newcommand{\e}{\varepsilon}

\newcommand\pos[2][]{_{[\ifthenelse{ \equal{#1}{} }{}{#1;}#2]}}

\newcommand{\pref}{\mathrm{Pref}}

\newcommand{\bfop}{\Phi}   
\newcommand{\sbfop}{\Phi^{\mathbf{s}}}   
\newcommand{\anybfop}{\Phi^{?}}
\newcommand{\lop}{\mathrm{Op}}  

\newcommand{\PSPACE}{\ensuremath{\mathrm{PSPACE}}\xspace}

\newcommand{\ree}{\bar{e}}  
\newcommand{\rea}{\bar{a}}  
\newcommand{\trt}{\ensuremath{\mathbf{t}}\xspace}   
\newcommand{\trs}{\ensuremath{\mathbf{s}}\xspace}   
\newcommand{\trta}{\ensuremath{\mathbf{t}^{\rea}}\xspace}   
\newcommand{\trti}{\ensuremath{\mathbf{t}^{-1}}\xspace}   

\newcommand{\auta}{\ensuremath{\mathbf{a}}\xspace}   
\newcommand\ML{\ensuremath{\auta_L}\xspace}
\newcommand\MtL{\ensuremath{\auta_{\theta(L)}}\xspace}

\newcommand{\id}{\mathrm{id}}
\newcommand{\idk}{{\mathrm{id}_k}}

\newcommand{\dna}{\ensuremath{\delta}\xspace}

\newcommand{\A}{\ensuremath{A}\xspace}
\newcommand{\D}{\ensuremath{\Delta}\xspace}        
\newcommand{\As}{\ensuremath{A^*}\xspace}
\newcommand{\Ap}{\ensuremath{A^+}\xspace}
\newcommand{\Ds}{\ensuremath{\D^*}\xspace}
\newcommand{\Dp}{\ensuremath{\D^+}\xspace}
\newcommand{\Ak}{\ensuremath{A_k}\xspace}
\newcommand{\Aks}{\ensuremath{A_k^*}\xspace}
\newcommand{\Akp}{\ensuremath{A_k^+}\xspace}
\newcommand{\Al}{\ensuremath{A_\ell}\xspace}
\newcommand{\Als}{\ensuremath{A_\ell^*}\xspace}
\newcommand{\Alp}{\ensuremath{A_\ell^+}\xspace}
\newcommand{\Abin}{\ensuremath{A_2}\xspace}

\newcommand{\Abinp}{\ensuremath{A_2^+}\xspace}
\newcommand{\Abinm}{\ensuremath{A_2^m}\xspace}

\newcommand{\Stt}[1][\trt]{\ensuremath{\cS_{\theta,#1}}\xspace}
\newcommand{\Wtt}[1][\trt]{\ensuremath{\cW_{\theta,#1}}\xspace}
\newcommand{\Wttt}[2]{\ensuremath{\cW_{#1,#2}}\xspace}
\newcommand{\SP}{\mbox{\texorpdfstring{\cS}{S}-property}\xspace}
\newcommand{\SPs}{\mbox{\texorpdfstring{\cS}{S}-properties}\xspace}
\newcommand{\WP}{\mbox{\texorpdfstring{\cW}{W}-property}\xspace}
\newcommand{\WPs}{\mbox{\texorpdfstring{\cW}{W}-properties}\xspace}
\newcommand\bfp{\cB}     
\newcommand\sbfp{\cB^{\mathbf{s}}}   
\newcommand\anybfp{\cB^{?}}

\newcommand\strict{{\mathbf{s}}}
\newcommand\weak{{\mathbf{w}}}

\newcommand\wildx{{\mathbf{x}}}

\newenvironment{transducer}[1][]
	{\begin{center}
	\begin{tikzpicture}[font=\footnotesize,shorten >=1pt,node distance=3.5cm,
	on grid,>=stealth',initial text=,
	every node/.style={align=center},
	every state/.append style={inner sep=1pt},
	every path/.style={->,bend angle=20},
	#1]}
	{\end{tikzpicture}\end{center}}

\newcommand{\breakingrow}[2][c]{%
  \begin{tabular}[#1]{@{}c@{}}#2\end{tabular}}

\newcounter{propcounter}
\renewcommand{\thepropcounter}{(\Alph{propcounter})}
\newlength{\propwidth}
\newlength{\proplabelw}

\newenvironment{dnaprop}
	{\setlength\propwidth{\textwidth}
	\addtolength\propwidth{-3em}
	\par\addvspace{1ex}	
	\noindent\hfill
	\begin{minipage}{\propwidth}
		\refstepcounter{propcounter}
		\settowidth\proplabelw{\thepropcounter\ }
		\hspace{-\proplabelw}\thepropcounter\ \unskip
	}
	{\end{minipage}\par\addvspace{1ex}}

\NewEnviron{journal}{\BODY}[]
\NewEnviron{conference}{}
\newcommand{\jcmath}[1]{\[#1\]}

\begin{document}

\title{Transducer Descriptions of DNA Code Properties and Undecidability of Antimorphic Problems}

\author{Lila Kari$^{1}$ \and Stavros Konstantinidis$^{2}$ \and Steffen Kopecki$^{1,2}$}
\date{}

\maketitle

\begin{journal}
\begin{center}
\small
$^{1}$ The University of Western Ontario, London, Ontario, Canada\\
\texttt{lila@csd.uwo.ca}, \texttt{steffen@csd.uwo.ca}
\\
$^{2}$ Saint Mary's University, Halifax, Nova Scotia, Canada\\
\texttt{s.konstantinidis@smu.ca}
\end{center}
\bigskip
\end{journal}

\begin{abstract}
This work concerns formal descriptions of \emph{DNA} code properties, and builds on previous
work on transducer descriptions of \emph{classic} code properties and
on trajectory descriptions of \emph{DNA} code properties.
This line of research 
allows us to give a property as input to an
algorithm, in addition to any regular language, which can then answer questions about the language
and the property. Here we define DNA code properties via transducers and show that this method is
strictly more expressive than that of trajectories, without sacrificing the efficiency of deciding
the satisfaction question. We also show that the maximality question can be undecidable.
Our undecidability results hold not only for the fixed DNA involution but also for
any fixed antimorphic permutation. Moreover, we also show the undecidability of the
antimorphic version of the Post Corresponding Problem, for any fixed antimorphic permutation.
\end{abstract}

\section{Introduction}\label{sec:intro}
The study of \textit{formal} methods for describing independent language properties (widely known as
code properties) provides tools that allow one to give a property as input to an algorithm and answer
questions about this property. Examples of such properties include
\textit{classic} ones \cite{Shyr:Thierrin:relations,Shyr:book,Jurg:Konst:handbook,BePeRe:2009}
like prefix codes, bifix codes, and
various error-detecting languages, as well as DNA code properties
\cite{Baum:1996,kari2002codes,HuKaKo:2003,KKLW:2003,MaFe:2004,JoMaCh:2005,KaKoSo:2005,KaKoSo:2005a,JoKaMa:2008,GeMa:2012,FaWaHu:2014}
like $\theta$-nonoverlapping and $\theta$-compliant languages.
A formal description method should be expressive enough to allow one to describe many desirable properties. Examples of
formal methods for describing \textit{classic} code properties are the implicational conditions method of \cite{Jurg:1999},
the trajectories method of \cite{Dom:2004}, and the transducer methods of \cite{DudKon:2012}.  The latter two
have been implemented to some extent in the Python package FAdo \cite{FAdo}.
A formal method for describing DNA code properties is
the method of trajectory DNA code properties \cite{KaKoSo:2005,Dom:2007}.

Typical questions about properties  are the following:
\begin{description}
\item{\em Satisfaction problem:}
given the description of a property and the description of a regular language, decide whether the language satisfies the property.
\item{\em Maximality problem:}
given the description of a property and the description of a regular language that satisfies the
property, decide whether the language is maximal with respect to the given property.
\item{\em Construction problem:}
given the description of a property and a positive integer $n$, find  a language of $n$ words (if possible)
satisfying the given property.
\end{description}

In the above problems  regular languages are described via (non-deterministic) finite automata (NFA).
Depending on the context, properties are described via trajectory regular expressions
or transducer expressions.
The satisfaction problem is the most basic one and can be answered usually efficiently in polynomial time.
The maximality problem as stated above can be decidable, in which case it is normally PSPACE-hard.
For existing transducer properties, both problems can be answered using the online (formal) language server LaSer \cite{Laser}, which relies on FAdo.
\begin{journal}
LaSer allows users to enter the desired property and language, and returns either the answer in real time
(online mode), or it returns a Python program that computes the desired answer if executed at the user's site
(program generation mode).
\end{journal}
For the construction problem a simple statistical algorithm is included
in FAdo, but we think that this problem is far from being well-understood.
\pnsi
The \emph{general objective} of this research is to develop methods for formally describing
DNA code properties that would allow
one to  express various combinations of such properties and be able to get answers to questions about
these properties.
\begin{journal}
While the satisfaction and construction questions are important from both the theoretical
and practical viewpoints, the maximality question is at least of theoretical interest and a classic problem
in the theory of codes.
\end{journal}
The contributions of this work are as follows:
\begin{enumerate}
\item
The definition of a  new simple formal method for describing many DNA code properties,
called $\theta$-transducer properties, some of which cannot be described by the existing
transducer and trajectory methods for classic code properties; see Sect.~\ref{sec:theta-transducers}.
\begin{journal}
These methods are closed under intersection of code properties.
This  means that if two properties
can be described within the method then also the combined property can be described within
the method. This outcome is important as in practice it is desirable that languages satisfy
more than one property.
\end{journal}
\item
The demonstration that the new method of transducer DNA code properties is properly more expressive than the method of
trajectories; see Sect.~\ref{sec:expressiveness}.
\item
The demonstration that the maximality problem can be decidable for some transducer DNA
code properties but undecidable for some others; see Sect.~\ref{sec:dec-problems}.
\item
The demonstration that some  classic undecidable problems (like PCP) remain undecidable when rephrased in terms of  {\em any fixed} \am permutation $\theta$
of the alphabet, with  the case $\theta = \id$ corresponding to these classic problems, where $\id$ is the (morphic) identity; see Sect.~\ref{sec:PCP}.
\begin{journal}
This contribution is mathematically relevant to the undecidability of the maximality
problem for DNA-related properties, so it is natural to include it with the above contributions  in one
publication.
\end{journal}
\begin{journal}
\item
The presentation of a natural hierarchy of DNA properties which are all \theta-transducer properties; see Section~\ref{sec:dna-properties}.
This hierarchy generalizes the concept of bond-free properties in \cite{kari2002codes,HuKaKo:2003,KKLW:2003}.
\end{journal}
\end{enumerate}

Even though, our main motivation is the description of DNA-related properties, we follow the more general approach which considers properties described by transducers involving a fixed \am permutation \theta; again, the classical transducer properties are obtained by letting $\theta = \id$.
\begin{journal}
In the setting of DNA properties, we consider the alphabet $\Delta = \set{\bA,\bC,\bG,\bT}$ and $\theta=\dna$ being the involution (\ie antimorphic permutation with $\dna^2=\id$) given by $\dna(\bA) = \bT$, $\dna(\bT) = \bA$, $\dna(\bC) = \bG$, and $\dna(\bG) = \bC$.
As it turns out, in the case when $\theta$ is morphic all questions that we consider in this paper can be answered analogous to the solutions for the classical case where $\theta=\id$.
Therefore, we focus on the transducer properties involving antimorphic permutations in this paper.
\end{journal}

\section{Basic Notions and Background Information}\label{sec:notion}

In this section we lay down our notation for formal languages, \am permutations, transducers, and language properties.
We assume the reader to be familiar with the fundamental
concepts of language theory; see \eg \cite{HopcroftUllman,FLhandbook}.
Then, in Sect.~\ref{secTP} we recall the method of  transducers for describing classic code properties, and in Sect.~\ref{sec:trajectories} we recall the method of trajectories for describing DNA-related properties.

\subsection{Formal Languages and (Anti-)morphic Permutations}
\begin{journal}
An \emph{alphabet} \A is a finite set of \emph{letters};
\As is the set of all words or strings over \A;
$\e$~denotes the \emph{empty word};
and $\Ap = \As\sm \set\e$.
A \emph{language} $L$ over $A$ is a subset $L\sse \As$;
the \emph{complement $L^c$} of $L$ is the language $\As\sm L$.
For an integer $m\in\N$ we let $\A^{\leq m}$ denote the set of words whose length is at most $m$; \ie $\A^{\leq m} = \bigcup_{i\leq m}\A^i$.
The \emph{DNA alphabet} is  $\Delta=\{\bA, \bC, \bG, \bT\}$.
Often it is convenient to consider the \emph{generic alphabet $\Ak = \set{0,1,\ldots,k-1}$ of size $k$} rather than a general alphabet;
note that $\A_2\sse \A_3\sse \A_4\sse \cdots{}$.
Throughout this paper we only consider alphabets with at least two letters because our investigations would become trivial over unary alphabets.
\end{journal}

\begin{conference}
For an alphabet $\A$ and a language $L$ over $\A$ we have the notation:
$\Ap= \As\sm \set\e$, where $\e$ is the empty word; and $L^c = \As\sm L$.
For an integer $k\ge 2$ we define the \emph{generic alphabet $\Ak = \set{0,1,\ldots,k-1}$ of size $k$}.
Throughout this paper we only consider alphabets with at least two letters because our investigations would become trivial over unary alphabets.
\end{conference}

Let $w\in \As$ be a word.
Unless confusion arises, by $w$ we also denote the singleton language~$\set w$, \eg $L\cup w$
means $L\cup\{w\}$.
If $w = xyz$ for some $x,y,z\in \As$, then $x$, $y$, and $z$ are called \emph{prefix}, \emph{infix} (or \emph{factor}), and \emph{suffix} of $w$, respectively.
For a language $L\sse \As$, the set $\pref(L) = \sett{x\in \As}{\exists y\in \As\colon xy\in L}$ denotes the language containing all prefixes of words in $L$.
If $w = a_1a_2\cdots a_n$ for letters $a_1,a_2,\ldots, a_n\in \A$, then $\abs w = n$ is the \emph{length of $w$};
for $b\in \A$, $\abs{w}_b = \abs{\sett{i}{a_i = b,1\le i\le n }}$ is the tally of $b$ occurring in $w$;
the $i$-th letter of $w$ is $w\pos{i} = a_i$ for $1\le i\le n$;
the infix of $w$ from the $i$-th letter to the $j$-th letter is $w\pos[i]{j} = a_ia_{i+1}\cdots a_{j}$ for $1\le i \le j \le n$;
and the \emph{reverse} of $w$ is $w^R = a_na_{n-1}\cdots a_1$.

Consider a generic alphabet $\Ak$ with $k\ge 2$.
The \emph{identity} function on $\Ak$ is denoted by $\idk$; when the alphabet is clear from the context, the index $k$ is omitted.
For a \emph{permutation} (or bijection) $\theta\colon\Ak\to \Ak$,
\begin{journal}
the permutation $\theta^{-1}$ is the \emph{inverse} of $\theta$ as usual; \ie $\theta\circ\theta^{-1} = \idk$ (``$\circ$'' is the composition of two functions $(g\circ h)(x) = g(h(x))$ for all $x$).
For $i\in \Z$,
\end{journal}
\begin{conference}
 and for $i\in\Z$,
\end{conference}
the permutation $\theta^i$ is the \emph{$i$-fold composition} of $\theta$; \ie $\theta^0 = \idk$, $\theta^{i} = \theta\circ\theta^{i-1}$, and $\theta^{-i} = (\theta^i)^{-1} = (\theta^{-1})^i$ for $i > 0$.
\begin{conference}
An \emph{involution} $\theta$ is a permutation such that $\theta = \theta^{-1}$.
\end{conference}
\begin{journal}
There exists a number $n$, called the \emph{order} of $\theta$, such that $\theta^n = \id_k$.
An \emph{involution} $\theta$ is a permutation of order $2$; \ie $\theta = \theta^{-1}$.
\end{journal}

A permutation $\theta$ over $\Ak$ can naturally be extended to operate on words in $\Aks$ as%
\begin{inparaenum}[(a)]
\item
\emph{morphic permutation} $\theta(uv) = \theta(u)\theta(v)$, or
\item
\emph{antimorphic permutation} $\theta(uv) = \theta(v)\theta(u)$, for $u,v\in\Aks$.
\end{inparaenum}
As before, the inverse $\theta^{-1}$ of the \am permutation $\theta$ over $\Aks$ is the \am extension of the permutation $\theta^{-1}$ over $\Aks$.
\begin{journal}
Note that the composition of two antimorphic or two morphic permutations yields a morphic permutation, whereas the composition of a morphic and an antimorphic permutation yields an antimorphic permutation.
Therefore, if $\theta$ is an antimorphic permutation, then $\theta^i$ is morphic if and only if $i$ is even.
\end{journal}
The identity $\idk$ always denotes the morphic extension of $\idk$ while the antimorphic extension of $\idk$, called the \emph{mirror image} or reverse, is usually denoted by the exponent $^R$.

\begin{example}
The \emph{DNA involution}, denoted as $\dna$, is an antimorphic involution on
$\D=\{\bA,\bC,\bG,\bT\}$ such that
$\dna(\bA)=\bT$ and $\dna(\bC)=\bG$, which implies $\dna(\bT)=\bA$ and $\dna(\bG)=\bC$.
\end{example}

\begin{journal}
A \emph{language operator} is any mapping  $\lop\colon2^{\As}\to 2^{\As}$.
The prefix function $\pref$ defined earlier is an example
of a language operator.  A transducer (see Sect.~\ref{secTP}) can be viewed
as a language operator.
Any \am permutation, as well as any other function, $h\colon \As \to \As$ over words is extended
to a language operator such that
for $L\sse \As$
\[
	h(L) = \cup_{x\in L}\{h(x)\}.
\]

If $\lop_1$ and $\lop_2$ are language operators, then $(\lop_1\lor\lop_2)$
is the language operator such that $(\lop_1\lor\lop_2)(X)=\lop_1(X)\cup\lop_2(X)$, for all languages $X$.
\end{journal}

\subsection{Describing Classic Code Properties by Transducers}\label{secTP}

A \emph{(language) property} \cP is any set of languages.
A language $L$ \emph{satisfies} \cP, or has \cP, if $L\in\cP$.
Here by a property \cP we mean an \emph{($n$-)independence} in the sense of \cite{Jurg:Konst:handbook}:
there exists $n\in\N\cup\set{\aleph_0}$
such that a language $L$ satisfies \cP if and only if all nonempty subsets $L'\sse L$ of cardinality less than $n$ satisfy \cP.
A language $L$ satisfying \cP is \emph{maximal} (with respect to \cP) if for every word $w\in L^c$ we have $L\cup w$ does not satisfy \cP---note that, for any  independence \cP, every language in \cP is a subset of a maximal language in \cP \cite{Jurg:Konst:handbook}.
\begin{journal}
To our knowledge all code related properties in the literature, including DNA code properties, are independence properties. 
\end{journal}
As we shall see further below the focus of this work is on 3-independence
properties that can also be viewed as independent with respect to a binary relation in the
sense of \cite{Shyr:Thierrin:relations}.

A \emph{transducer} \trt is a non-deterministic finite state automaton with output; see \eg \cite{Be:1979,Yu:handbook}.
\begin{journal}
In general, a transducer can have an output alphabet $B$ which is different from its input alphabet $\A$; thus, defining a relation over $\A^*\times B^*$.
In this paper however, we only consider transducers where the input alphabet coincides with the output alphabet, $\A = B$, which leads to the following simplified definition:
\end{journal}
\begin{conference}
Here we only consider transducers whose input and output alphabets are equal:
\end{conference}
a transducer is a quintuple $\trt = (Q,\A,E,I,F)$, where $\A$ is the input and output alphabet,
$Q$ is a finite set of states, $E$ is a set of directed edges between states from $Q$ which are labeled by word pairs $(u,v) \in \As\times \As$,
$I$ is a set of initial states, and $F$ a set of final states.
\begin{journal}
For an edge label $(u,v)$ the word $u$ is called {\em input}, while the word $v$ is called {\em output}.
The transducer $\trt$ \emph{realizes} the set of all pairs $(x,y)\in \As\times \As$ such that $x$ is formed by concatenating the inputs, and $y$ is formed by concatenating the outputs of the labels in
a path of $\trt$ from the initial to the final states.
\end{journal}
If \trt realizes $(x,y)$ then we write
$y\in\trt(x)$. We say that the set $\trt(x)$ contains all possible outputs of $\trt$ on input $x$.
\begin{journal}
It is well known that for two regular languages $R_1,R_2$ there exists a transducer $\trt$ that realizes the relation $R_1 \times R_2$; \ie $\trt$ realizes $(x,y)$ if and only if $x\in R_1$ and $y\in R_2$.
\end{journal}
The transducer $\trti$ is the inverse of $\trt$; that is, $x\in\trti(y)$ if and only if $y\in\trt(x)$ for all words $x,y$.
\begin{journal}
Note that $\trti$ is obtained from $\trt$ by simply swapping the input with the output word on each edge in \trt.
For a language $L$ we naturally extend our notation such that
\begin{align*}
	\trt(L) &= \cup_{x\in L}\trt(x). 
\end{align*}
Thus, a transducer can be viewed as a language operator.
\pnsi
\end{journal}
Let \theta be an \am permutation and \trt be a transducer which are both defined over the same alphabet \A.
The transducer \trt is called \emph{\theta-input-preserving} if for all $w\in\Ap$ we have $\theta(w) \in \trt(w)$; \trt is called  \emph{\theta-input-altering} if for all $w\in \Ap$ we have $\theta(w)\notin\trt(w)$.
We use the simpler terms \emph{input-altering} and
\emph{input-preserving} \trt, respectively, when
$\theta=\id$.
Note that $\theta(w) \in \trt(w)$ is equivalent to $w\in \thetai(\trt(w))$ as well as $\trti(\theta(w)) \ni w$.

%
%
\begin{definition}[\cite{DudKon:2012}]\label{def:existing}
 An input-altering transducer $\trt$ \emph{describes the property} that consists of all languages $L$
such that
\begin{equation}\label{eqIATP}
\trt(L)\cap L=\es.
\end{equation}
An input-preserving transducer $\trt$ \emph{describes the property} that consists of all languages $L$
such that
\begin{equation}\label{eqIPTP}
w\notin \trt(L\sm w),\>\>\hbox{ for all $w\in L$}.
\end{equation}
A property is called an input-altering (resp. input-preserving) transducer property, if
it is described by an input-altering (resp. input-preserving) transducer.
\end{definition}

Note that every input-altering transducer property is also an input-preserving transducer property.
Input-altering transducers can be used to describe properties like prefix codes, bifix codes,
and hypercodes. Input-preserving transducers are intended for error-detecting properties,
where in fact the transducer plays the role of the communication channel.
\begin{journal}
Figure~\ref{trans:classic} shows a couple of examples.

\begin{figure}[ht!]
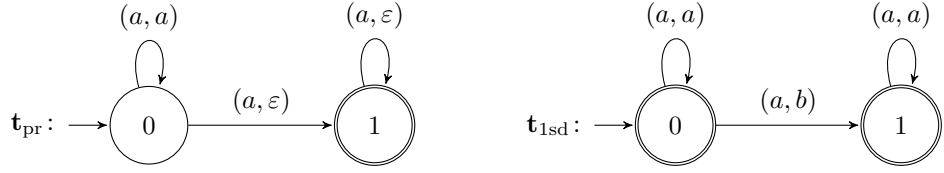

\begin{transducer}[node distance=3cm]
	\node [state,initial] (q0) {$0$};
	\node [node distance=1cm,left=of q0,anchor=east] {$\trt_{\rm pr}\colon$};
	\node [state,accepting,right of=q0] (q1) {$1$};
	\node [state,initial,accepting,right of=q1,node distance=4cm] (q2) {$0$};
	\node [node distance=1cm,left=of q2,anchor=east] {$\trt_{\rm 1sd}\colon$};
	\node [state,accepting,right of=q2] (q3) {$1$};
	\path (q0) edge [loop above] node [above] {$(a,a)$} ()
		(q0) edge node [above] {$(a,\e)$} (q1)
		(q1) edge [loop above] node [above] {$(a,\e)$} ()	
		(q2) edge node [above] {$(a,b)$} (q3)
		(q2) edge [loop above] node [above] {$(a,a)$} ()
		(q3) edge [loop above] node [above] {$(a,a)$} ();
\end{transducer}
\caption{The left transducer is input-altering and describes the prefix codes: on input $x$ it
outputs any proper prefix of $x$. The right transducer is input-preserving and describes the
1-substitution error-detecting languages: on input $x$ it outputs either $x$ or any word differing from
$x$ in exactly one position.
\underline{Note:} in this and the following transducer figures, an arrow
with label $(a,a)$ represents a set of edges with labels $(a,a)$ for all $a\in\A$; and similarly for an arrow
with label  $(a,\e)$. An arrow with label $(a,b)$ represents
a set of edges with labels $(a,b)$ for all $a,b\in\A$ with $a\not=b$.}
\label{trans:classic}
\end{figure}
\end{journal}
Many  input-altering transducer properties can be described in a simpler manner
by \emph{trajectory regular expressions} \cite{Dom:2004,DudKon:2012}, that is, regular expressions over \{0, 1\}. For example, the
expression $0^*1^*$ describes prefix codes and the expression $1^*0^*1^*$
describes infix codes. On the other hand, there
are  natural transducer properties that cannot be described by
trajectory expressions \cite{DudKon:2012}.

\subsection{Describing DNA-related Properties by Trajectories}
\label{sec:trajectories}

In \cite{Baum:1996,kari2002codes,HuKaKo:2003,KKLW:2003,MaFe:2004,JoMaCh:2005,KaKoSo:2005,KaKoSo:2005a,JoKaMa:2008,GeMa:2012,FaWaHu:2014} the authors consider numerous properties of languages inspired by reliability issues in DNA computing.
We state three of these properties below.
\begin{journal}
In Sect.~\ref{sec:dna-properties} we present a hierarchy
of DNA properties which generalizes
some of the DNA properties presented in~\cite{kari2002codes,HuKaKo:2003,KKLW:2003}.
\end{journal}
Let $\theta$ be an antimorphic permutation over $\Aks$.
Recall that in the DNA setting $\theta=\dna $ is an involution, and therefore, we have $\theta^2 = \id$.

\begin{dnaprop}
\label{dna:nol}
A language $L$ is \emph{$\theta$-nonoverlapping} if $L\cap\theta(L) = \es$.
\end{dnaprop}

\begin{dnaprop}
$L$ is \emph{$\theta$-compliant} if
$\forall w\in \theta(L), x,y\in\Aks\colon xwy\in L \implies xy=\e$.
\label{dna:compliant}
\end{dnaprop}

\begin{dnaprop}
$L$ is \emph{strictly $\theta$-compliant} if it is
$\theta$-nonoverlapping and $\theta$-compliant.
\label{dna:Scompliant}
\end{dnaprop}

Many of the existing DNA-related properties  can be modelled using the concept
of a bond-free property, first defined in \cite{KaKoSo:2005} and later rephrased in
\cite{Dom:2007} in terms of trajectories. We follow the fomulation in \cite{Dom:2007}.
Let $\ree=(\ree_1,\ree_2)$, where $\ree_1$ and $\ree_2$ are two regular trajectory expressions.
First, we define the following language operators.
\begin{eqnarray}
\bfop_{\ree}(L) &=& (((L\sdel_{\ree_1}\Ap)\cap\Ap)\sins_{\ree_2}\As)
\cup(((L\sdel_{\ree_1}\As)\cap\Ap)\sins_{\ree_2}\Ap).\label{eqBFOP}\\
\sbfop_{\ree}(L) &=& ((L\sdel_{\ree_1}\As)\cap\Ap)\sins_{\ree_2}\As.\label{eqSBFOP}
\end{eqnarray}

\begin{conference}
The language operations $\sins_t$ and $\sdel_t$ are \emph{shuffle} (or scattered insertion) and \emph{scattered deletion},
respectively, over the trajectory $t$; see \cite{KaSo:2005,Dom:2007} for details.
\end{conference}

\begin{journal}
The word operations $\sins_t$ and $\sdel_t$ are called \emph{shuffle} (or scattered insertion) and
\emph{scattered deletion},
respectively, over the trajectory $t$.
A \emph{trajectory} is any word over $\{0,1\}$.
For any words $x,w$ and trajectory $t$ with $|t|_0=|x|$ and $|t|_1=|w|$,
$x\sins_t w$ is the set $\{y\}$ such that the word $y$ is of length $|t|$
and results by the following process which
scans the symbols of $x$ left to right and also of $w$ left to right.
For each index $i=0,\ldots,|t|-1$, $y[i]$ is the next symbol of $x$ if $t[i]=0$,
or the next symbol of $w$ if $t[i]=1$. If `$|t|_0=|x|$ and $|t|_1=|w|$' is
not satisfied then $x\sins_t w=\es$.
For example, $1122\sins_{001010}34=113242$.
The reader is referred to \cite{KaSo:2005,Dom:2007} for more details.
For any languages $X,W$ and trajectory expression $\rea$, we  have that
\[
X\sins_{\rea} W = \bigcup_{x\in X,w\in W,t\in L(\rea)} x\sins_t w.
\]
\pnsi
For any words $x,w$ and trajectory $t$ with $|t|=|x|$ and $|t|_1=|w|$,
$x\sdel_t w$ is either the set $\{y\}$ such that the word $y$ is of length
$|t|_0=|x|-|w|$ and satisfies $\{x\}=y\sins_t w$, or the empty set otherwise.
For example, $113242\sdel_{001010}34=1122$.
The reader is referred to \cite{KaSo:2005,Dom:2007} again for more details.
For any languages $X,W$ and trajectory expression $\rea$, we  have that
\[
X\sdel_{\rea} W = \bigcup_{x\in X,w\in W,t\in L(\rea)} x\sdel_t w.
\]
\end{journal}

\begin{definition}\label{defTraj} [\cite{Dom:2007}]
Let \theta be an involution and $\ree_1,\ree_2$ be two regular trajectory expressions.
The  \emph{bond-free property described by} $(\ree_1,\ree_2)$ is
\begin{equation}\label{eqBFP}
\bfp{(\ree_1,\ree_2)}=\{L\sse\As\mid \theta(L)\cap \bfop_{\ree_1,\ree_2}(L)=\es\}.
\end{equation}
The   \emph{strictly bond-free property described by} $(\ree_1,\ree_2)$ is
\begin{equation}\label{eqSBFP}
\sbfp{(\ree_1,\ree_2)}=\{L\sse\As\mid \theta(L)\cap \sbfop_{\ree_1,\ree_2}(L)=\es\}.
\end{equation}
A \emph{regular $\theta$-trajectory  property} is
a  bond-free property described by $(\ree_1,\ree_2)$,
or a  strictly bond-free property described by $(\ree_1,\ree_2)$,
for some pair $(\ree_1,\ree_2)$.
\end{definition}

\begin{journal}
\begin{example}
The \theta-compliant property is a regular $\theta$-trajectory property in
$\bfp{(1^*0^+1^*,0^+)}$: deleting $x$ and $y$ in any $xwy$ (according to $1^*0^+1^*$),
where at least one symbol gets deleted,
and then inserting nothing (according to $0^+$) cannot result into a word in $\theta(L)$.
The \theta-nonoverlapping property is a regular $\theta$-trajectory property in
$\sbfp{(0^+,0^+)}$: deleting nothing and then inserting nothing in any word $w$  cannot
result into a word in $\theta(L)$.
The strictly \theta-compliant property is a regular $\theta$-trajectory property in
$\sbfp{(1^*0^+1^*,0^+)}$: deleting $x$ and $y$ in any $xwy$ (according to $1^*0^+1^*$)
and inserting nothing (according to $0^+$) cannot result into a word in $\theta(L)$.
\end{example}
\end{journal}
\begin{journal}
We note that the actual definitions of bond-free properties in \cite{Dom:2007}
are given in terms of a pair $(T_1,T_2)$ of arbitrary sets of trajectories. However, here we only
consider  sets of trajectories that can be represented by regular expressions.
Moreover,  the second statement of Theorem~\ref{thExpress}, in Sect.~\ref{sec:expressiveness},
remains true if one uses $(T_1,T_2)$ instead of $(\ree_1,\ree_2)$, as the proof makes
no use of the fact that the trajectory sets involved are regular.
\end{journal}

\section{New Transducer-based DNA-related Properties}
\label{sec:theta-transducers}

A question that arises from the discussion in sections~\ref{secTP} and~\ref{sec:trajectories} is whether existing transducer-based
properties include DNA-related properties.
It turns out that this is not the case%
\begin{journal}%
: for instance the
$\dna$-nonoverlapping property, which seems to be the simplest DNA-related property,
cannot be described by any input-preserving transducer
\end{journal}%
; see Proposition~\ref{propNO}.
In this section, we define new transducer-based properties that are
appropriate for DNA-related applications, we demonstrate Proposition~\ref{propNO}, and
discuss how existing DNA-related properties can be described with transducers.
\begin{journal}
Then, in Sect.~\ref{sec:expressiveness} we examine the relationship between
the new transducer properties and the regular $\theta$-trajectory properties which were proposed in \cite{Dom:2007}.
\end{journal}

\begin{definition}\label{def:properties}
A transducer \trt and an \am permutation \theta, defined over the same alphabet, describe $3$-independent properties in two ways:
\begin{enumerate}[1.)]
\item
\emph{strict \theta-transducer property (\SP)}:
$L$ satisfies the property \Stt if
\begin{equation}\label{eqPP}\theta(L)\cap\trt(L) = \es\end{equation}
\item
\emph{weak \theta-transducer property (\WP)}:
$L$ satisfies the  property \Wtt if
\begin{equation}\label{eqIP}\forall w\in L\colon \theta(w) \notin\trt(L\sm w)\end{equation}
\end{enumerate}
Any of the properties \Stt or \Wtt is called a \emph{\theta-transducer property}.
\end{definition}

The difference between \SPs and \WPs is that \Stt forbids that $L\in\Stt$ contains a word $w$ such that any $\theta(w)\in \trt(w)$, while this case is allowed for  $L\in\Wtt$.
For fixed \trt, \theta, and $L$, Condition~\eqref{eqPP} implies that 
for all $w\in L$ we have $\theta(w)\cap \trt(L\sm w) = \es$ which is equivalent to Condition~\eqref{eqIP}.
In other words, if $L$ satisfies $\Stt$, then $L$ satisfies $\Wtt$ as well. 
If $\theta=\id$ and $\trt$ is input-altering, or input-preserving, then the above defined properties specialize to the existing ones stated in Definition~\ref{def:existing}.

\begin{figure}[ht!]
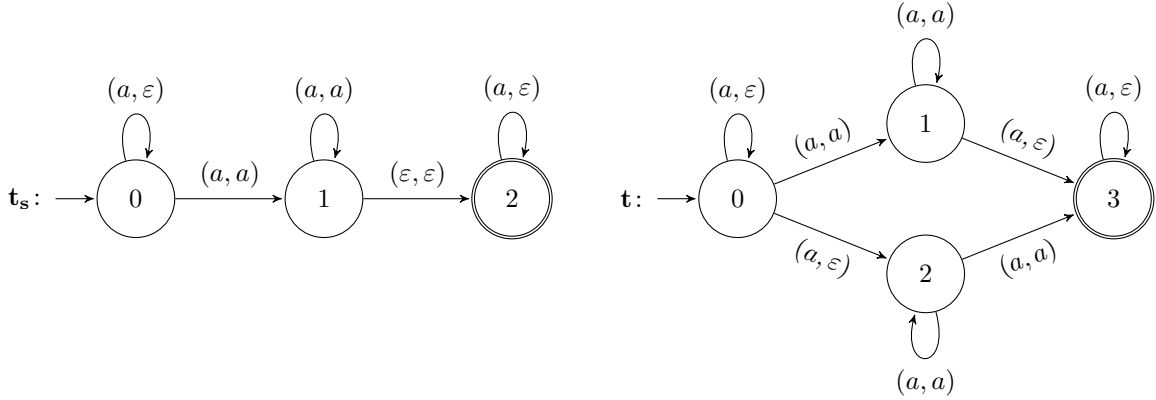

\begin{conference}
\begin{transducer}[node distance=.35cm and 2cm]
	\node [state,initial] (p0) {$0$};
	\node [node distance=.75cm,left=of p0,anchor=east] {$\trt\colon$};
	\node [state,above right=of p0] (p1) {$1$};
	\node [state,below right=of p0] (p2) {$2$};
	\node [state,accepting,below right=of p1] (p3) {$3$};	

	\path [sloped] (p0) edge [loop above] node [above] {$(a,\e)$} ()
		(p0) edge node [above] {$(a,a)$} (p1)
		(p1) edge [loop above] node [above] {$(a,a)$} ()	
		(p1) edge node [above] {$(a,\e)$} (p3)
		(p0) edge node [below] {$(a,\e)$} (p2)
		(p2) edge [loop below] node [below] {$(a,a)$} ()
		(p2) edge node [below] {$(a,a)$} (p3)
		(p3) edge [loop above] node [above] {$(a,\e)$} ();

	\begin{scope}[xshift=-8cm]
	\node [state,initial] (q0) {$0$};
	\node [node distance=.75cm,left=of q0,anchor=east] {$\trt_{\strict}\colon$};
	\node [state,right=of q0] (q1) {$1$};
	\node [state,accepting,right=of q1] (q2) {$2$};

	\path (q0) edge [loop above] node [above] {$(a,\e)$} ()	
		(q0) edge node [above] {$(a,a)$} (q1)
		(q1) edge [loop above] node [above] {$(a,a)$} ()	
		(q1) edge node [above] {$(\e,\e)$} (q2)
		(q2) edge [loop above] node [above] {$(a,\e)$} ();
	\end{scope}
\end{transducer}
\end{conference}
\begin{journal}
\begin{transducer}[node distance=1cm and 2.5cm]
	\node [state,initial] (p0) {$0$};
	\node [node distance=1cm,left=of p0,anchor=east] {$\trt\colon$};
	\node [state,above right=of p0] (p1) {$1$};
	\node [state,below right=of p0] (p2) {$2$};
	\node [state,accepting,below right=of p1] (p3) {$3$};	

	\path [sloped] (p0) edge [loop above] node [above] {$(a,\e)$} ()
		(p0) edge node [above] {$(a,a)$} (p1)
		(p1) edge [loop above] node [above] {$(a,a)$} ()	
		(p1) edge node [above] {$(a,\e)$} (p3)
		(p0) edge node [below] {$(a,\e)$} (p2)
		(p2) edge [loop below] node [below] {$(a,a)$} ()
		(p2) edge node [below] {$(a,a)$} (p3)
		(p3) edge [loop above] node [above] {$(a,\e)$} ();

	\begin{scope}[xshift=-8cm]
	\node [state,initial] (q0) {$0$};
	\node [node distance=1cm,left=of q0,anchor=east] {$\trt_{\strict}\colon$};
	\node [state,right=of q0] (q1) {$1$};
	\node [state,accepting,right=of q1] (q2) {$2$};

	\path (q0) edge [loop above] node [above] {$(a,\e)$} ()	
		(q0) edge node [above] {$(a,a)$} (q1)
		(q1) edge [loop above] node [above] {$(a,a)$} ()	
		(q1) edge node [above] {$(\e,\e)$} (q2)
		(q2) edge [loop above] node [above] {$(a,\e)$} ();
	\end{scope}
\end{transducer}
\end{journal}
\caption{Together with \theta, the left transducer describes the strictly \theta-compliant property
and the right one describes the \theta-compliant property. See Example~\ref{exTransd}
for explanations.}
\label{trans:dna-example}
\end{figure}
\begin{conference}
\vspace{-0,5cm}
\end{conference}

\begin{example}\label{exTransd}
\begin{conference}
In Fig.~\ref{trans:dna-example}, an arrow
with label $(a,a)$ represents a set of edges with labels $(a,a)$ for all $a\in\A$; and similarly for an arrow
with label  $(a,\e)$.
\end{conference}
\begin{journal}
Consider the transducers in Fig.~\ref{trans:dna-example}.
\end{journal}
For any word $xwy$, the left transducer $\trt_\strict$, say, can delete $x$, then keep $w$ (which has to be non-empty), and then delete $y$.
Thus, $\trt_{\strict}(L)\cap\theta(L)=\es$ if and only if $L$ is strictly \theta-compliant.
Now let $xwy$ with $xy\neq \e$ and $w\neq \e$.
If $y$ is nonempty, the right transducer $\trt$ can delete $x$, then keep $w$, and then delete $y$ using the upper path (containing state~$1$);
and if $x$ is nonempty, $\trt$ can delete $x$, then keep $w$, and then delete $y$ using the lower path (containing state~$2$).
Thus, $\trt(L)\cap\theta(L)=\es$ if and only if $L$ is \theta-compliant.
Using FAdo \cite{FAdo} format the left transducer can be specified by
the following string, assuming alphabet \{\texttt{a}, \texttt{b}\}
\pssn
\verb#@Transducer 2 * 0\n0 a @epsilon 0\n0 b @epsilon 0\n0 a a 1\n#
\pnsn
\verb#0 b b 1\n1 a a 1\n1 b b 1\n1 @epsilon @epsilon 2\n2 a @epsilon 2\n#
\pnsn
\verb#2 b @epsilon 2\n#
\end{example}

\begin{journal}
As in the classic case where $\theta=\id$, also in the general case we have that \theta-input-altering transducers play an important role for \SPs because only then the maximality question is decidable.
We did not fully explore the usefulness of \theta-input-preserving for antimorphic permutations yet.
For morphic \theta, however, every transducer \trt can be modified to obtain a \theta-input-preserving transducer $\trt'$ such that $\Wtt = \Wtt[\trt']$;
this concept can be utilized in order to efficiently decide the satisfaction problem; see Sect.~\ref{sec:dec-problems}.

\begin{remark}\label{rem:one-word}
Note that only \Stt for a transducer \trt which is not $\theta$-input-altering can exclude specific words from all languages which satisfy the property \Stt.
Otherwise, when $\trt$ is \theta-input-altering, it must not realize $(w,\theta(w))$;
and when we consider an \WP, then $\theta(w)\in \trt(w)$ is allowed for $w\in L$.
In particular, every singleton language $L = \set{w}$ satisfies all properties $\Wtt$, as well as, $\Stt$ if \trt is \theta-input-altering.
\end{remark}
\end{journal}

\begin{journal}
As input-altering transducer properties are a subset of input-preserving transducer properties, we only consider the case of input-preserving transducer properties in the next two results.
\end{journal}

The next result  demonstrates that existing transducer properties are not suitable for describing even simple
DNA-related properties.

\begin{proposition}\label{propNO}
The $\dna$-nonoverlapping property is not describable by any input-preserving transducer.
\end{proposition}

\begin{journal}
\begin{proof}
The singleton language $L = \set{\bA\bT}\sse\Ds$ is not $\dna$-nonoverlapping, because
the word $\bA\bT = \dna(\bA\bT)$ is a \dna-palindrome.
Analogously to Remark~\ref{rem:one-word}, a  transducer property $P_\trt = \Wttt{\id}{\trt}$, which is described
by some input-preserving transducer \trt, cannot exclude any singleton language.
Therefore, we must have $L\in P_\trt$.
\end{proof}

The counter example language $\set{\bA\bT}$ used to prove the previous result is rather artificial,
as in practice code-related languages should have more than two elements.
However, the statement remains true even if we focus on languages containing more than one word.
This case is handled in the next proposition.

\begin{proposition}\label{propNO2}
There is no input-preserving transducer \trt that
satisfies Equation~\eqref{eqIPTP} for all $\dna$-non\-over\-lapping languages $L$ having at least two elements.
\end{proposition}

\begin{proof}
Assume the contrary, that is, there is an input-preserving transducer $\trt$ such that for any
DNA language $L\sse\Ds$ with at least two element we have
\[
\dna(L)\cap L=\es\quad\hbox{iff}\quad \forall \,u\in L:\>\trt(u)\cap (L\sm u)=\es.
\]
We can assume that \trt is in normal form, that is, the label of every edge is of the form $(a,\e)$
or $(\e,a)$, for some $a\in\D$.
Assume that \trt has $n$ states, for some positive integer $n$, and let $m>n$. We have that
$\{\bA^m\bC^m,\bG^m\bT^m\}$ is not $\dna$-nonoverlapping, so  without loss
of generality we have that $\bG^m\bT^m\in\trt(\bA^m\bC^m)$. Consider an
accepting path $\pi$ of \trt
whose label is $(\bA^m\bC^m,\bG^m\bT^m)$ and say $\pi$ consists of $N$ consecutive edges, for some positive integer $N$. Then, these edges are
$s_{i-1}\xras{(x_i,y_i)}s_i$, for $i=1,\ldots,N$, so that the concatenation of the $x_i$'s is equal to
$\bA^m\bC^m$ and the concatenation of the $y_i$'s is equal to $\bG^m\bT^m$.
As \trt is in normal form, we have $N=4m$, and as $m>n$, there is a smallest integer
$k\ge1$ such that state $s_k$ is equal to a previous one, that is $s_k=s_{j}$  such that
$j<k.$
By the choice of $k$, we have $k\le n<m$.
Let $x=x_1\cdots x_j$, $u=x_{j+1}\cdots x_k$, $x'=x_{k+1}\cdots x_N$, and
$y=y_1\cdots y_j$, $v=y_{j+1}\cdots y_k$, $y'=y_{k+1}\cdots y_N$.
As $j-k>0$ and \trt is in normal form we have that
\begin{equation}\label{eqLemmaNO}
|u|>0\quad \hbox{or}\quad |v|>0.
\end{equation}
Using a standard pumping
argument for finite state machines, the path that results if we delete from $\pi$
the $k-j$ edges between $s_j$ and $s_k$ is also an accepting path
whose label is $(xx',yy')$. As each $x_i$ and $y_i$ is of
length 0 or 1, we have $|xu|\le k<m$ and  $|yv|<m$, and also $|u|\le k-j$ and $|v|\le k-j$.
This implies $xx'=\bA^{m-|u|}\bC^m$ and $yy'=\bG^{m-|v|}\bT^m$. As $xx'\not=yy'$
and $yy'\in\trt(xx')$ we have that $\{xx',yy'\}$ is not  $\dna$-nonoverlapping, which
implies $xx'=\dna(yy')$, that is, $\bA^{m-|u|}\bC^m=\bA^{m}\bC^{m-|v|}$ and,
therefore, $|u|=|v|=0$ which contradicts~(\ref{eqLemmaNO}).
\end{proof}
\end{journal}

\section{Expressiveness of Transducer-based Properties}
\label{sec:expressiveness}
In this section we examine the descriptive power of the newly defined transducer
DNA-related properties, that is, the $\theta$-transducer properties.
In Theorem~\ref{thExpress} we show that these properties properly
include the regular \theta-trajectory properties.
On the other hand, in Proposition~\ref{propNotTransd}
we show that there is an  independent DNA-related property that is not a
\theta-transducer property.

\begin{proposition}\label{propNotTransd}
The \theta-free property (defined below) \cite{HuKaKo:2003} is not a \theta-transducer property.
\begin{dnaprop}
A language $L\sse \As$ is \emph{\theta-free} if and only if
$L^2 \cap \Ap \theta(L) \Ap = \es$.
\end{dnaprop}
\end{proposition}

\begin{journal}
\begin{proof}
First note that every \theta-transducer property is 3-independent, so it is sufficient to show that, for $\theta=\dna$ and $\A=\D$, the \theta-free
property is not 3-independent. Assume the contrary and consider the language
\[
K = \{\texttt{ACGT, CCAC, GTAA}\}.
\]
This is not $\dna$-free, as  $\texttt{ACGT}=\dna(\texttt{ACGT})$
and $\texttt{CCACGTAA} \in \Dp\texttt{ACGT}\Dp$.
On the other hand, one verifies that every nonempty subset of $K$ of cardinality less than 3
is $\dna$-free, so by our assumption also $K$ must be $\dna$-free, which
is a contradiction.
\end{proof}
\end{journal}


\begin{journal}
The remainder of this section is devoted to Theorem~\ref{thExpress}.
Recall the DNA alphabet is $\D=\{\bA,\bC,\bG,\bT\}$.
\end{journal}
The following DNA language property is considered in Theorem~\ref{thExpress}
\[
\cH =\{L\sse\Ds\mid H(u,\theta(v))\ge2,\>\hbox{ for all $u,v\in L$}\},
\]
where $H(\cdot,\cdot)$ is the Hamming distance function with the assumption
that its value is $\infty$ when applied on different length words.
Note that \cH is described by $\dna$ and the transducer
shown in Fig.~\ref{trans:pH}.

\begin{figure}[ht!]
\begin{conference}
\begin{transducer}[node distance=2.5cm]
	\node [state,initial,accepting] (q2) {$0$};
	\node [state,accepting,right of=q2] (q3) {$1$};
	\path 		(q2) edge node [above] {$(a,b)$} (q3)
		(q2) edge [loop above] node [above] {$(a,a)$} ()
		(q3) edge [loop above] node [above] {$(a,a)$} ();
\end{transducer}
\end{conference}
\begin{journal}
\begin{transducer}[node distance=3cm]
	\node [state,initial,accepting] (q2) {$0$};
	\node [state,accepting,right of=q2] (q3) {$1$};
	\path 		(q2) edge node [above] {$(a,b)$} (q3)
		(q2) edge [loop above] node [above] {$(a,a)$} ()
		(q3) edge [loop above] node [above] {$(a,a)$} ();
\end{transducer}
\end{journal}
\caption{The  transducer describing, together with $\dna$, the \SP $\cH$.}
\label{trans:pH}
\end{figure}

\begin{example}\label{exLs}
\begin{conference}
The DNA language $L_1=\{\bA\bG\bG,\> \bC\bC\bA\}$ does not satisfy $\cH$ because $H(\bC\bC\bA,\dna(\bA\bG\bG))=1$.
The DNA language $L_2=\{\bA\bA\bA,\> \bC\bC\bT\}$ satisfies \cH because $\dna(\bA\bA\bA)=\bT\bT\bT$ and all words $u\in L_2$ contain at most
one $\bT$. 
\end{conference}

\begin{journal}
The following DNA languages do not satisfy $\cH$:
\pmsi
$L_0=\{\bA\bG\bG,\> \bC\bC\bA\},\quad L_0'=\{\bG\bA\bG,\>\bC\bC\bC\}.$
\pmsn
For instance, $H(\bC\bC\bA,\dna(\bA\bG\bG))=1$.
The following languages satisfy $\cH$:
\pmsi
$L_1=\{\bA\bC\bG,\> \bG\bA\bT\},\quad L_2=\{\bC\bA\bC,\>\bG\bC\bT\},$
\pssi
$L_3=\{\bA\bA\bA,\> \bC\bC\bT\},\quad L_4=\{\bA\bA\bA,\> \bC\bT\bC\}
,\quad L_5=\{\bA\bA\bA,\> \bT\bC\bC\}.$
\pmsn
For instance, as $\dna(\bA\bA\bA)=\bT\bT\bT$ and all words $u\in L_3$ contain at most
one $\bT$, it follows that $H(u,\dna(\bA\bA\bA))\ge2$. Now using $\dna(\bC\bC\bT)=\bA\bG\bG$,
one verifies that  $H(u,\dna(\bC\bC\bT))\ge2$ for any $u\in L_3$. Thus, indeed $L_3$
satisfies $\cH$.
\end{journal}
\end{example}

\begin{theorem}\label{thExpress}\
\begin{enumerate}
\item
Let \theta be an antimorphic involution. Every regular \theta-trajectory  property
is a \theta-transducer property.
\item
Property $\cH$ is a $\dna$-transducer property, but not a
(regular) $\dna$-trajectory one.
\end{enumerate}
\end{theorem}

\begin{journal}
\begin{proof}
We use the following notation:
$\anybfop_{\ree}$ for either of the operators $\bfop_{\ree}$ and $\sbfop_{\ree}$,
and
$\anybfp(\ree)$ for either of the properties $\bfp(\ree)$ and $\sbfp(\ree)$.

For the \underline{first} statement, we show that given any trajectory regular expression $\rea$,
each of the following operators is a transducer operator
\begin{eqnarray*}
\trta_1(X) &=& X\sdel_{\rea}\As\\
\trta_2(X) &=& X\sins_{\rea}\As\\
\trta_3(X) &=& X\sdel_{\rea}\Ap\\
\trta_4(X) &=& X\sins_{\rea}\Ap
\end{eqnarray*}
The statement then would follow by noting that if \trt and \trs are transducer operators
then also $(\trt\circ\trs)$ and $(\trt\lor\trs)$ are transducer operators \cite{Be:1979},
and if $\auta$ is an automaton, then one can construct  the transducer $(\trs\uar\auta)$  such that $y\in(\trs\uar\auta)(x)$ if and only if
$y\in\trs(x)\cap L(\auta)$ \cite{Kon:2002}.
For example, for any pair $\ree=(\ree_1,\ree_2)$, we have that
\[
\sbfop_{\ree}(L)=(\trt^{\ree_2}_2\circ(\trt^{\ree_1}_1\uar\auta_+))(L),
\]
where $\auta_+$ is any automaton accepting $\Ap$.

The claim about  $\trta_4$ is already shown in \cite{DudKon:2012}.
For the claim about $\trta_2$, first note that
$X\sins_{\rea}\As=(X\sins_{\rea}\Ap)\,\cup\,(X\sins_{\rea}\{\e\})$,
so  $\trta_2$ is equal to $(\trta_4\lor\trt_{\rea,\mathrm{id}})$, where $\trt_{\rea,\mathrm{id}}$
is a transducer with $\trt_{\rea,\mathrm{id}}(x)=x\sins_{\rea}\{\e\}$ and defined as follows.
First note that by definition, $y\in x\sins_{\rea}\{\e\}$ if and only if $y=x$ and $0^{|x|}\in L(\rea)$.
Let $\auta$ be an automaton with no empty transitions accepting $L(\rea)$.
Then, $\trt_{\rea,\mathrm{id}}$ is made based on $\auta$ as follows. Its set  of transitions
consists of all tuples $(p,a/a,q)$ such that $(p,0,q)$ is a transition of $\auta$---we say that the latter
is the \emph{corresponding} transition of the former. The initial and final states
of $\trt_{\rea,\mathrm{id}}$ are those initial and final states, respectively, of $\auta$ that appear
in the transitions of $\trt_{\rea,\mathrm{id}}$. It follows that $\trt_{\rea,\mathrm{id}}$ realizes
a pair $(x,y)$ of words using some path $P$ of transitions, if and only if $x=y$ and the automaton $\auta$ accepts $0^{|x|}$ using a path consisting of the corresponding transitions that make the path $P$.

In \cite{KaSo:2005} it is observed that $y\in(x\sdel_t w)$ if and only if
$x\in(y\sins_t w)$, for all words $x,y,w$ and trajectories $t$, which implies that
$\trta_3$ and $\trta_1$ are simply the inverses of the transducers $\trta_4$ and $\trta_2$,
respectively.
\pssi
For the \underline{second} statement we recall that \cH is described by $\dna$ and the transducer
shown in Fig.~\ref{trans:pH}.
For the second part of the statement, we argue by contradiction, so we assume that there is a pair
of trajectory regular expressions $\ree=(\ree_1,\ree_2)$ such that
\[
\cH=\anybfp(\ree_1,\ree_2).
\]
Using the definition of $\anybfop$, one verifies that
\[
\anybfop(a)\sse a\As,\>\>\hbox{for all $a\in\A$}.
\]
Consider the DNA language $K=\{\texttt{A}, \texttt{C}\}$. One verifies that
$K$ does not satisfy $\cH$, but on the other hand $\dna(K)\cap\anybfop_{\ree}(K)=\es$,
which means that $K$ satisfies $\anybfp(\ree_1,\ree_2)$, which leads to the required contradiction.
\end{proof}

The counter example used to prove the second statement of Theorem~\ref{thExpress} is a little
artificial, as the language $K=\{\texttt{A}, \texttt{C}\}$ consists of 1-letter words, which
is of no practical value. The next result gives a stronger  statement, as it requires that all
words involved are of length at least 2.

\begin{proposition}\label{thExpress2}
The following property
\[
\cH_2 =\{L\sse\Ds\mid |u|\ge2\>\hbox{ and }\>H(u,\theta(v))\ge2,\>\hbox{ for all $u,v\in L$}\}
\]
is a $\dna$-transducer property but not a $\dna$-trajectory property.
\end{proposition}
The proof of this results require a couple of intermediate results, which we present next.
\begin{lemma}\label{lemExpress2}
Let $x,y$ be any words and $s,t$ be any trajectories.
If $y\in((x\sdel_s\As)\cap\Ap)\sins_t\As$ then
\[
|t|-|s|=|t|_1-|s|_1=|y|-|x|\quad\hbox{and}\quad |s|_1<|x|.
\]
\end{lemma}
\begin{proof}
The premise of the statement implies that $y\in z\sins_t w_2$
and $z\in((x\sdel_s w_1)\cap\Ap)$ for some words $z,w_1,w_2$
with $|z|>0$.
Informally, this means that $y$ results by deleting $|w_1|$ symbols from $x$,
with $|w_1|<|x|$,
and then inserting $|w_2|$ symbols. More formally as $|t|=|y|$ and $|s|=|x|$,
we have that $|t|-|s| = |y|-|x|$. Also as
$|z|=|x|-|w_1|=|s|-|s|_1$, we have that $|s|>|s|_1$
and, therefore, $|x|>|s|_1$, as required. Now, we have
\[
|s|_1=|w_1|=|x|-|z|=|x|-(|y|-|w_2|) = |x|-|y|+|t|_1
\]
and, therefore, $|t|_1-|s|_1=|y|-|x|$.
\end{proof}

\begin{lemma}\label{lemExpress}
Let $\ree=(\ree_1,\ree_2)$ be a pair of trajectory regular expressions and assume
that $\cH=\anybfp(\ree)$---as we shall see further below  this assumption leads to
a contradiction.
\begin{enumerate}
\item
There is no pair $(s,t)$ of trajectories in $L(\ree_1)\times L(\ree_2)$
such that $|s|=|t|=3$ and $|s|_1=|t|_1=2$.
\item
If $x,y$ are DNA words of length 3 and $(s,t)\in L(\ree_1)\times L(\ree_2)$ such that
$x\not=\dna(y)$ and $y\in((x\sdel_s\Ds)\cap\Dp)\sins_t\Ds$ then
$|s|=|t|=3$ and $|s|_1=|t|_1=1$.
\item
We have that $010\in L(\ree_1)$ or $010\in L(\ree_2)$.
\item
We have that $(001,001)\in L(\ree_1)\times L(\ree_2)$ or $(100,100)\in L(\ree_1)\times L(\ree_2)$.
\end{enumerate}
\end{lemma}

\begin{proof}
We shall use some of the seven languages in Example~\ref{exLs}.
\pnsi
For the \underline{first} statement, assume for the sake of contradiction that the two trajectories
have equal length and exactly two 1s each. 
By applying  $(\texttt{AAA}\sdel_{s}\Ds)\cap\Dp$  followed by $\sins_{t}\Ds$,  the result
is $\anybfop(\texttt{AAA})$ and is equal to $\bA\D\D$ or $\D\bA\D$ or $\D\D\bA$, depending
on whether $t=011$ or $t=101$ or $t=110$, respectively.
More specifically, if
$t=011$ then $\anybfop(\texttt{AAA})$ contains $\dna(\texttt{CCT})$, which contradicts the
fact that $L_3$ satisfies $\cH$.
If $t=101$ then $\anybfop(\texttt{AAA})$ contains $\dna(\texttt{CTC})$, which contradicts the
fact that $L_4$ satisfies $\cH$. If
$t=110$ then $\anybfop(\texttt{AAA})$ contains $\dna(\texttt{TCC})$, which contradicts the
fact that $L_5$ satisfies $\cH$.
\pssi
For the \underline{second} statement, Lemma~\ref{lemExpress2} implies that
$|s|=|t|=3$ and $|s|_1=|t|_1\le1$, and  $x\not=\dna(y)$ implies that $|s|_1\not=0$.
Hence, $|s|_1=|t|_1=1$, as required.
\pssi
For the \underline{third} statement, the fact that $L_0'$ does not satisfy $\cH$ implies that
there are words $u,v\in L_0'$ such that $\dna(v)\in\anybfop(u)$ and, therefore, there are
words $w_1,w_2$ and  $(s,t)\in L(\ree_1)\times L(\ree_2)$ such that
\[
\dna(v)\in((u\sdel_{s} w_1)\cap\Dp)\sins_{t} w_2.
\]
By the previous statement, $|s|=|t|=3$ and $|s|_1=|t|_1=1$,
which implies $|w_1|=|w_2|=1$.
For the sake of contradiction
assume $s\not=010$ and $t\not=010$. Let $u=u_1u_2u_3$ with each $u_i$ being a symbol.
There are four cases about the values of $s$ and $t$, all of which lead to contradictions.
For example, if $s=001$ and $t=001$ then $\dna(v)=u_1u_2w_2$, which implies that
$v=\bar w_2\bar u_2\bar u_1$. By inspection, one verifies that $u_1u_2u_3,\bar w_2\bar u_2\bar u_1$
cannot be both in $L_0'$.
\pssi
For the \underline{fourth} statement, the fact that $L_0$ does not satisfy $\cH$ implies that
there are words $u,v\in L_0$ such that $\dna(v)\in\anybfop(u)$ and, therefore, there are
words $w_1,w_2$ and  $(s,t)\in L(\ree_1)\times L(\ree_2)$ such that
\[
\dna(v)\in((u\sdel_{s} w_1)\cap\Dp)\sins_{t} w_2.
\]
By a previous statement, $|s|=|t|=3$ and $|s|_1=|t|_1=1$,
which implies $|w_1|=|w_2|=1$. Let $u=u_1u_2u_3$ with each $u_i$ being a symbol.
The rest of the proof consists of four parts:
\pssi $s=010$ leads to a contradiction;
\pnsi $t=010$ leads to a contradiction;
\pnsi $s=001$ implies $t=001$;
\pnsi $s=100$ implies $t=100$.
\pssn
We demonstrate the first and fourth parts and leave the other two parts to the reader to verify.
For the first part, if $s=010$ then depending on whether $t=001$ or $t=010$ or $t=100$,
we have that
$\dna(v)=u_1u_3w_2$ or $\dna(v)=w_2u_1u_3$ or $\dna(v)=w_2u_1u_3$, and hence,
$v=\bar w_2\bar  u_3\bar u_1$ or $v=\bar w_2\bar u_3\bar u_1$ or
$v=\bar u_3\bar u_1\bar w_2$. One verifies by inspection that, in any case, it is impossible
to have $u,v\in L_0$. Finally for the last part, if $s=100$ then, as $t$ cannot be 010, we
have that $\dna(v)=u_2u_3w_2$ or $\dna(v)=w_2u_2u_3$ and hence,
$v=\bar w_2\bar  u_3\bar u_2$ or $v=\bar u_3\bar u_2\bar w_2$. One verifies by inspection that, in
either case, it is impossible to have $u,v\in L_0$.
\end{proof}

\begin{proof} (Of Proposition~\ref{thExpress2}.)
The fact that $\cH_2$ is a $\dna$-transducer \SP is established using the transducer
in Fig.~\ref{trans:H}.
\begin{figure}[ht!]
\begin{transducer}[node distance=2cm and 2.5cm]
	\node [state,initial,accepting] (q0) {$0$};
	\node [state,accepting,below=of q0] (q1) {$1$};
	\node [state,accepting,below right=of q0] (q2) {$2$};
	\node [state,accepting,above right=of q2] (q3) {$3$};
	\path 		
	    (q0) edge node [left] {$(a,a)$}(q1)
	    (q0) edge node [below,sloped] {$(a,a)$} (q2)
	    (q0) edge node [above] {$(a,b)$} (q3)
	    (q2) edge node [below,sloped] {$(a,b)$} (q3)
		(q2) edge [loop right] node [right] {$(a,a)$} ()
		(q3) edge [loop right] node [right] {$(a,a)$} ();
\end{transducer}
\caption{The  transducer describing, together with $\dna$, the \SP $\cH_2$.}
\label{trans:H}
\end{figure}
For the second part of the statement, we argue by contradiction, so we assume that there is
a pair of trajectory regular expressions $(\ree_1,\ree_2)$ such that
\[
\cH_2=\anybfp(\ree_1,\ree_2).
\]
By Lemma~\ref{lemExpress}, we have that $001\in L(\ree_2)$ or $100\in L(\ree_2)$,
and that $001\in L(\ree_1)$ or $100\in L(\ree_1)$.
Moreover, we can distinguish the following four cases, which all lead to contradictions.
We also consider the languages $L_1$ and $L_2$ defined in Example~\ref{exLs}.
\pssn
\emph{Case  `$\,010\in L(\ree_1)$ and $001\in L(\ree_2)$'.}
Then, $\texttt{GCT}$ results into  $\texttt{GT}$, then into $\texttt{GTG}$ and then into
$\texttt{CAC}$ using, respectively, the operations $\sdel_{010}$, $\sins_{001}$ and $\dna$,
which contradicts the fact that $L_2$ satisfies $\cH$.
\pssn
\emph{Case  `$\,010\in L(\ree_1)$ and $100\in L(\ree_2)$'.}
Then, $\texttt{GAT}$ results into  $\texttt{GT}$, then into $\texttt{CGT}$ and then into
$\texttt{ACG}$ using, respectively, the operations $\sdel_{010}$, $\sins_{100}$ and $\dna$,
which contradicts the fact that $L_1$ satisfies $\cH$.
\pssn
\emph{Case  `$\,001\in L(\ree_1)$ and $010\in L(\ree_2)$'.}
Then, $\texttt{ACG}$ results into  $\texttt{AC}$, then into $\texttt{ATC}$ and then into
$\texttt{GAT}$ using, respectively, the operations $\sdel_{001}$, $\sins_{010}$ and $\dna$,
which contradicts the fact that $L_1$ satisfies $\cH$.
\pssn
\emph{Case  `$\,100\in L(\ree_1)$ and $010\in L(\ree_2)$'.}
Then, $\texttt{CAC}$ results into  $\texttt{AC}$, then into $\texttt{AGC}$ and then into
$\texttt{GCT}$ using, respectively, the operations $\sdel_{100}$, $\sins_{010}$ and $\dna$,
which contradicts the fact that $L_2$ satisfies $\cH$.
\end{proof}
\end{journal}

\section{The Satisfaction and Maximality Problems}\label{sec:dec-problems}

For $\theta=\id$  and for input-altering and -preserving transducers the satisfaction and maximality problems are decidable \cite{DudKon:2012}.
In particular, for a regular language $L$ given via an automaton $\auta$, Condition~\eqref{eqIATP} can be decided in time $\Oh(|\trt||\auta|^2)$, where the function $|\cdot|$ returns the size of the machine in question (its number of edges plus the length of all labels on the edges).
Condition~\eqref{eqIPTP} can be decided in time $\Oh(|\trt||\auta|^2)$, as noted in Remark~\ref{rem:input-alt-complexity}.
The maximality problem is decidable, but PSPACE-hard,
for both input-altering and -preserving transducer properties.

\begin{remark}\label{rem:input-alt-complexity}
Let $\trs = \trt\dar\auta\uar\auta$ be the transducer obtained by two product constructions: first on the input of \trt with \auta; then, on the output of the resulting transducer with \auta.
In \cite{DudKon:2012} the authors suggest to decide whether or not $L$ satisfies the input-preserving transducer property $\Wttt{\id}{\trt}$ by testing if the transducer $\trs$ is functional%
\begin{journal}
\ ($\abs{\trs(x)} \le 1$ for all $x\in\As$)
\end{journal}.
However, deciding $L\in\Wttt{\id}{\trt}$ can be done by the cheaper test of whether or not $\trs$ implements a (partial) identity function%
\begin{journal}
\ ($\trs(x) = \set x$ or $\trs(x) = \es$ for all $x\in\As$)
\end{journal}.
Using the identity test from \cite{AllMoh:2003}, we obtain that Condition~\eqref{eqIPTP} can be decided in time $\Oh(|\trt||\auta|^2)$ when the alphabet is considered constant.
Also note that the identity test does not require that $\trt$ is input-preserving if $\theta=\id$.
When $\theta$ is antimorphic, however, the identity test does not work anymore and we have to resort to the more expensive functionality test for $\theta$-input-preserving transducers.
\end{remark}

In this work we are interested in the case when $\theta\neq \id$ is antimorphic; furthermore, the \theta-input-altering or -preserving restrictions on the transducer are not necessarily present in the definition of $\WPs$ or $\SPs$.
Table~\ref{tab:un-decidability} summarizes under which conditions the satisfaction and maximality
problems are decidable for regular languages. 
\begin{journal}
For the satisfaction problem, except for the case of non-restricted transducer $\WPs$, Conditions~\eqref{eqPP} and~\eqref{eqIP} can be tested similarly to Conditions~\eqref{eqIATP} and~\eqref{eqIPTP}.
For the
case of non-restricted transducer $\WPs$, we show decidability using a different method; see
Sect.~\ref{sec:satisfaction:Itt}.
The undecidability result holds for every fixed permutation \theta over an alphabet with at least two letters, in particular, all results apply to the DNA-involution~\dna.
All maximality results are discussed in Sect.~\ref{sec:maximality}.
\end{journal}

\begin{table}[ht!]
\centering\small
\begin{tabular}{|c||c|c|c|c|}
\hline
& \multicolumn{2}{c|}{} & \multicolumn{2}{c|}{}\\[-2ex]
\multirow{2}{*}{Problem} &
\multicolumn{2}{c|}{Property \Stt}&\multicolumn{2}{c|}{Property \Wtt} \\
& no restriction & \trt is $\theta$-i.-altering
& no restriction & \trt is $\theta$-i.-preserving \\
\hline\hline
& \multicolumn{2}{c|}{} & &\\[-2ex]
Satisfaction & \multicolumn{2}{c|}{\breakingrow{decidable in $\Oh(|\trt||\auta|^2)$ \\
as in~\cite{DudKon:2012}}} &
\breakingrow{decidable \\
Theorem~\ref{thm:decide:I}
} &
\breakingrow{decidable in $\Oh(|\trt|^2|\auta|^4)$\\
as in~\cite{DudKon:2012}
} \\
\hline
& & \multicolumn{3}{c|}{}\\[-2ex]
Maximality & \breakingrow{undecidable\\
Corollary~\ref{cor:undecidable:maximality}} &
\multicolumn{3}{c|}{\breakingrow{decidable, \PSPACE-hard\\
Theorem~\ref{thm:maximality}, Corollary~\ref{cor:max:pspace}}}\\
\hline
\end{tabular}
\bigskip

\caption{(Un-)decidability of the satisfaction and the maximality problems for a fixed antimorphic permutation $\theta$, a given transducer \trt, and a  regular language $L$ given via an automaton $\auta$.}
\label{tab:un-decidability}
\end{table}

\begin{remark}
We note  that deciding the satisfaction question for any \theta-trajectory property involves testing the emptiness conditions in \eqref{eqBFP} or \eqref{eqSBFP}, which requires time $\Oh(|\auta|^2|\auta_1||\auta_2|)$, where $\auta_1,\auta_2$ are automata corresponding to $\ree_1,\ree_2$. Such a property can be expressed as \theta-transducer \SP (recall Theorem~\ref{thExpress}) using a transducer of size $\Oh(|\auta_1||\auta_2|)$ and, therefore, the satisfaction question
can still be solved within the same asymptotic time complexity.
\end{remark}

\subsection{The Satisfaction Problem for non-restricted \WPs}\label{sec:satisfaction:Itt}

We establish the decidability of non-restricted transducer $\WPs$ for regular languages.
We do not concern the complexity of this algorithm; optimizing the algorithm and analyzing its complexity is part of future research.
Let \trt be a transducer, \theta be an antimorphic permutation, and $L$ be a regular language over the alphabet \A.
Let \ML and \MtL be the NFAs accepting the languages $L$ and $\theta(L)$, respectively.
Let
$\trs
\begin{journal}
 = (Q_\trs,\A,E_\trs,I_\trs,F_\trs)
\end{journal}
 = \trt\dar\ML\uar\MtL$
be the product transducer such that
$y\in \trs(x)$ if and only if $y \in \trt(x)$, $x\in L$, and $y\in \theta(L)$.
\begin{journal}
We consider \trs to be \emph{trim}, \ie every state in $Q_\trs$ lies on a path that leads from an initial state to a final sate.
Furthermore, \trs is considered to be in \emph{normal form} such that every edge is either labeled $(a,\e)$ or $(\e,a)$ for some letter $a\in\A$.
Thus, for any path $p \xras{(x,y)} q$ of length $\ell$ (the path has $\ell$ edges) in $\trs$ we have $\abs{xy} = \ell$.

\begin{lemma}\label{lem:Ittpath}
Let $L$ be a regular language, \trt be a transducer, \theta be an antimorphic involution, and $\trs = \trt\dar\ML\uar\MtL$ (all defined over \A).
The regular language $L$ satisfies \Wtt if and only if for all words $x,y\in\Ap$
\[
	y\in \trs(x) \implies \theta(x) = y.
\]
\end{lemma}

\begin{proof}
We will prove the contrapositive:
$L\notin \Wtt$ if and only if there exists $x,y\in\Ap$ such that $y\in \trs(x)$ and $\theta(x) \neq y$.
Recall that $L\notin \Wtt$ if and only if there exists $w\in L$ such that $\theta(w)\in \trt(L\sm w)$.

Assume that $L\notin \Wtt$ and, therefore, $w\in L$ exists such that $\theta(w)\in \trt(L\sm w)$.
Let $x \in L\sm w$ such that $\theta(w)\in \trt(x)$ and $y = \theta(w)\in \theta(L)$.
Clearly, we have $y\in \trs(x)$ and $y \neq \theta(x)$.

Conversely, assume that $x,y\in\Ap$ exists such that $y\in \trs(x)$ and $y\neq\theta(x)$.
Let $w = \thetai(y)$ and note that $w \in L$ (because $y\in\theta(L)$), $x\in L\sm w$, and $\theta(w)\in \trt(x)\subseteq\trt(L\sm w)$.
Therefore, $L\notin \Wtt$.
\end{proof}
\end{journal}

Let $T_\trs = \sett{(x_1,x_2,x_3)\in(\As)^3}{\abs{x_1x_2x_3}\le\abs\trs}$ be a set of word triples.
Note that the length restrictions for the words ensures that $T_\trs$ is a finite set.
For each triple $t = (x_1x_2x_3)\in T_\trs$ we define a relation
\[
	R_t = \sett{(x_1(x_2)^kx_3,\theta(x_1(x_2)^kx_3))}{k\in \N}
	 \sse \As\times\As.
\]
\begin{journal}
Note that we allow that any word of $x_1,x_2,x_3$ is empty; in particular, if $x_2=x_3=\e$, then $R_t$ contains only one pair of words $(x_1,\theta(x_1))$.
\end{journal}

\begin{lemma}\label{lem:good:triples}
\begin{journal}
Let $L$ be a regular language, \trt be a transducer, \theta be an antimorphic involution, and $\trs = \trt\dar\ML\uar\MtL$ (all defined over \A).
\end{journal}
The regular language $L$ satisfies $\Wtt$ if and only if the relation realized by \trs
\begin{conference}
is included in $\>\bigcup_{t\in T_\trs} R_t$.
\end{conference}
\begin{journal}
satisfies
\begin{equation}
	\trs \sse \bigcup_{t\in T_\trs} R_t.\label{eq:good:triples}
\end{equation}
\end{journal}
\end{lemma}

\begin{journal}
\begin{proof}
Recall that for every $(x,y) \in R_t$ with $t\in T_\trs$ we have $\theta(x) = y$.
If \trs satisfies Equation~\eqref{eq:good:triples}, then for all $(x,y)$ which are realized by \trs, we have $\theta(x) = y$;
and by Lemma~\ref{lem:Ittpath} $L$ satisfies \Wtt.

Conversely, suppose that $L$ satisfies \Wtt, let $(x,y)$ be a pair of words that is realized by \trs, and note that $\theta(x) = y$ by Lemma~\ref{lem:Ittpath}.
If $\abs{x} \le \abs\trs$, then $(x,\theta(x))=(x,y) \in R_t$ for $t = (x,\e,\e)\in T_\trs$.

Otherwise, every accepting path in \trs that is labeled by $(x,\theta(y))$ contains more than $\abs\trs$ edges, and therefore, must have a repeating state $p$
\[
	s \xras{(x_1,y_1)} p \xras{(x_2,y_2)} p \xras{(x_3,y_3)} f
\]
such that $x=x_1x_2x_3$, $\theta(x)=y_1y_2y_3$, $s\in I_\trs$, $f\in F_\trs$, $x_2y_2\neq \e$, $\abs{x_1x_2y_1y_2} \le \abs{\trs}$
(using the pigeonhole principle).
By Lemma~\ref{lem:Ittpath} for all $i\in \N$
\[
	x_1x_2^i x_3 = \thetai(y_1y_2^iy_3) = \thetai(y_3)\thetai(y_2)^i\thetai(y_1).
\]
Firstly note, that this implies $\abs{x_2} = \abs{y_2}$.
Now, consider $i = 2\abs{x}$.
Because $\abs{x_1x_2x_3} \ge \abs\trs \ge \abs{x_1x_2y_1y_2}$, we have that $\thetai(y_2)\thetai(y_1)$ is a suffix of $x_3$.
Since $i$ is sufficiently large, the suffix $x_2x_3$ of $x_1x_2^i x_3$ cannot overlap with the prefix $\thetai(y_3)$ of $x_1x_2^i x_3$.
Hence, there exists a suffix $u$ of $\thetai(y_2)$ and an integer $j\ge 2$ such that
\[
	x_2x_3 = u\thetai(y_2)^j \thetai(y_1).
\]
Chose $v$ such that $\thetai(y_2) = vu$ and note that $x_2 = uv$ because $\abs{x_2} = \abs{y_2}$ (this argument is a special case of the well-known Fine and Wilf's Theorem).
Let $x_3' = u\thetai(y_1)$ and observe that $x_3 = u(vu)^{j-1} \thetai(y_1) = x_2^{j-1}x_3'$.
Furthermore, $\abs{x_1x_2x_3'} \le \abs{x_1x_2y_1y_2}\le \abs{\trs}$.
We conclude that $(x,\theta(x)) = (x_1x_2^{j}x_3',\theta(x_1x_2^{j}x_3')) \in R_t$ for $t = (x_1,x_2,x_3')\in T_\trs$.
\end{proof}

In order to test whether or not Equation~\eqref{eq:good:triples} is satisfied, we perform two separate tests.
Firstly, we test whether or not \trs satisfies the weaker condition
\begin{equation}
	\trs \sse \bigcup_{(x_1,x_2,x_3)\in T_s} ( x_1x_2^*x_3 ) \times
		\theta(x_1 x_2^* x_3).\label{eq:test1}
\end{equation}
Secondly, we ensure that
\begin{equation}
	\forall x,y\colon y\in \trs(x) \implies \abs x = \abs y.\label{eq:test2}
\end{equation}

\begin{lemma}\label{lem:two:tests}
Equation~\eqref{eq:good:triples} is satisfied if and only if Equations~\eqref{eq:test1} and~\eqref{eq:test2} are satisfied.
\end{lemma}

\begin{proof}
If Equation~\eqref{eq:good:triples} is satisfied, then Equation~\eqref{eq:test1} is satisfied because $R_{(x_1,x_2,x_3)} \sse ( x_1x_2^*x_3 ) \times
\theta(x_1 x_2^* x_3)$ for $(x_1,x_2,x_3)\in T_\trs$.
Also note that for all $(x,y)\in R_t$ with $t\in T_\trs$ we have $\abs x = \abs y$; therefore, Equation~\eqref{eq:good:triples} implies Equation~\eqref{eq:test2}.

Conversely, assume that Equations~\eqref{eq:test1} and~\eqref{eq:test2} are satisfied.
For all $(x,y)$ that are realized by \trs we have there exists $(x_1,x_2,x_3)\in T_\trs$ and $i,j\in\N$ such that $x = x_1x_2^i x_3$ and $y = \theta(x_1 x_2^j x_3)$.
Since the equation $\abs x = \abs y$ must also be satisfied, it is clear that $i = j$ and, hence, $(x,y)\in R_{(x_1,x_2,x_3)}$.
We conclude that Equations~\eqref{eq:test1} and~\eqref{eq:test2} imply Equation~\ref{eq:good:triples}.
\end{proof}
\end{journal}

\begin{conference}
The inclusion in Lemma~\ref{lem:good:triples} is decidable performing the following two tests:
\begin{inparaenum}[1.)]
\item
verify that
$\trs\sse \bigcup_{(x_1,x_2,x_3)\in T_s} ( x_1x_2^*x_3 ) \times	\theta(x_1 x_2^* x_3)$; and
\item
verify that $\abs x = \abs y$ for all pairs $(x,y)$ that label an accepting path in \trs.
\end{inparaenum}
Note that the inclusion test can be performed because the right-hand-side relation is recognizable \cite{Be:1979}.
The second test follows the same ideas as the algorithm outlined in \cite{AllMoh:2003} which decides whether or not a transducer implements a partial identity function.
\end{conference}

\begin{theorem}\label{thm:decide:I}
Let $L$ be a regular language given as automaton, \trt be a given transducer, and \theta be a given antimorphic involution (all defined over \A).
It is decidable whether $L$ satisfies $\Wtt$ or not.
\end{theorem}

\begin{journal}
\begin{proof}
According to Lemmas~\ref{lem:good:triples} and~\ref{lem:two:tests} we have to decide whether or not the two Equations~\eqref{eq:test1} and~\eqref{eq:test2} are satisfied for the transducer $\trs = \trt\dar\ML\uar\MtL$.
It is known that it is decidable whether or not a given transducer is included in a recognizable relation (that is a relation $\bigcup_{i=1}^n A_i\times B_i$ for regular $A_i,B_i$); see \cite{Be:1979}.
Therefore, the inclusion in Equation~\eqref{eq:test1} is decidable.

The property in Equation~\eqref{eq:test2} can be verified by an algorithm that assigns an integer to each state in \trs:
the integer $i$ is assigned to $q\in Q_\trs$ if there exists a path $s\xras{(x,y)} q$ from a starting state $s\in I_\trs$ such that $i = \abs{x}-\abs{y}$.
The test fails if a state is assigned two distinct integers or if a final state from $F_\trs$ is assigned an integer different from $0$; otherwise, the test is successful.
Assigning the integers can be done by a simple depth-first traversal of \trs.
We omit further details on the implementation of this algorithm as it can be done analogously to the test whether or not a given transducer implements a (partial) identity function which can be found in \cite{AllMoh:2003}.
\end{proof}
\end{journal}

\subsection{The Maximality Problem}\label{sec:maximality}

Here we show how to decide maximality of a regular language $L$ with respect to a \theta-transducer property; see Theorem~\ref{thm:maximality}.
This result only holds when we consider \WPs or when we consider \SPs for \theta-input-altering transducers.
As in the case of existing transducer properties, it turns out that the maximality problem is \PSPACE-hard; see Corollary~\ref{cor:max:pspace}.
When we consider general \SPs, the maximality problem becomes undecidable; see
Corollary~\ref{cor:undecidable:maximality}.

\begin{theorem}\label{thm:maximality}
For an antimorphic permutation $\theta$, a transducer $\trt$, and a regular language $L$, all defined over $\Aks$, such that either
\begin{conference}
$L\in \Wtt$, or $L\in \Stt$ and \trt is $\theta$-input altering, we have that
\end{conference}
\begin{journal}
\begin{enumerate}[i.)]
\item
$L\in \Wtt$ or
\item
$L\in \Stt$ and \trt is $\theta$-input altering,
\end{enumerate}
\end{journal}
$L$ is maximal with property $\Wtt$ (\resp $\Stt$) if and only if
\begin{equation}
	L\cup \thetai(\trt(L)) \cup \trti(\theta(L)) = \Aks.
	\label{eq:maximality}
\end{equation}
\end{theorem}

\begin{journal}
\begin{proof}
\begin{inparaenum}[\em i.)\ ]
\item
Suppose $L\cup \thetai(\trt(L)) \cup \trti(\theta(L)) = \Aks$.
For every word $w\in L^c$ we have $\theta(w)\in\trt(L)$ or $w\in \trti(\theta(L))$.
In the former case, we immediately obtain that $L\cup w$ does not satisfy \Wtt.
In the latter case, there exists $u\in L$ such that $\theta(u) \in \trt(w)$, and therefore, $L \cup w$ does not satisfy \Wtt.
We conclude that $L$ is maximal with respect to \Wtt

Conversely, suppose there exists a word $w$ such that $w\notin L\cup \thetai(\trt(L)) \cup \trti(\theta(L))$.
Clearly, $w\in L^c$.
Furthermore, we must have $\theta(w)\notin \trt(L)$ and $\theta(u)\notin \trt(w)$ for all $u\in L$.
Since $L\in \Wtt$, we also have that $\theta(u)\notin \trt(L\sm u)$ for all $u\in L$.
Thus, we obtain that $\forall u\in (L\cup w) \colon \theta(u)\notin\trt((L\cup w)\sm u)$, and therefore, $L$ is not maximal with respect to \Wtt.

\item
Suppose $L\cup \thetai(\trt(L)) \cup \trti(\theta(L)) = \Aks$.
For all $w\in L^c$ we have $\theta(w) \cap \trt(L) \neq \es$ or $\trt(w) \cap \theta(L) \neq \es$.
Thus, $L \cup w$ does not satisfy \Stt and $L$ is maximal with respect to \Stt

Conversely, suppose there exists a word $w$ such that $w\notin L\cup \thetai(\trt(L)) \cup \trti(\theta(L))$.
Hence, $\theta(w) \cap \trt(L)  = \es$ and  $\trt(w) \cap \theta(L) = \es$.
Furthermore, we have $\theta(L) \cap \trt(L) = \es$ because $L\in\Wtt$ and $\theta(w)\cap \trt(w)=\es$ because $\trt$ is \theta-input-altering.
We conclude that $L\cup w$ satisfies \Wtt, and therefore, $L$ is not maximal with respect to \Wtt.
\end{inparaenum}
\end{proof}
\end{journal}

We note that it is \PSPACE-hard to decide whether or not Equation~\eqref{eq:maximality} holds when $L$ is given as NFA because it is \PSPACE-hard to decide universality of a regular language given as NFA ($L\sse\Aks$ is universal if $L = \Aks$) \cite{StMe1973}.

\begin{corollary}\label{cor:max:pspace}
For an antimorphic permutation $\theta$, a transducer $\trt$, and a regular language $L$ given as NFA, all defined over $\Aks$, such that either
\begin{conference}
$L\in \Wtt$, or $L\in \Stt$ and \trt is $\theta$-input altering, we have that
\end{conference}
\begin{journal}
\begin{compactenum}[i.)]
\item
$L\in \Wtt$ or
\item
$L\in \Stt$ and \trt is $\theta$-input altering,
\end{compactenum}
\end{journal}
it is \PSPACE-hard to decide whether or not $L$ is maximal with property $\Wtt$ (\resp $\Stt$).
\end{corollary}

\begin{journal}
\begin{proof}
According to Theorem~\ref{thm:maximality} deciding maximality of $L$ with property $\Wtt$ (\resp $\Stt$) is equivalent to deciding universality of $L\cup \thetai(\trt(L)) \cup \trti(\theta(L))$.
Let $\trt_\es$ be a transducer without final state which does not accept any pair of words.
Now, $L$ is maximal with property $\Stt[\trt_\es]$ (\resp $\Wtt[\trt_\es]$) if and only if $L$ is universal---a problem which is known to be \PSPACE-hard.
\end{proof}
\end{journal}

In the rest of this section we show that it is undecidable whether or not a transducer is $\theta$-input-preserving.
This question relates directly to the maximality problem of the empty language $\es$ with respect to the property \Stt, as stated in Corollary~\ref{cor:undecidable:maximality}.
\begin{conference}
The following Theorem can be proven using a reduction from the famous, undecidable Post correspondence problem (PCP) to the problem of deciding whether a given transducer is \theta-input-preserving or not.
\end{conference}
\begin{journal}
We will reduce the famous, undecidable Post correspondence problem to the problem of deciding whether or not a given transducer is \theta-input-preserving.

\begin{definition}\label{def:pcp}
The \emph{Post correspondence problem} (PCP) is, given words $\alpha_0,\alpha_1,\ldots,\alpha_{\ell-1} \in \Sigma^+$ and $\beta_0,\beta_1,\ldots,\beta_{\ell-1} \in \Sigma^+$, decide whether or not there exists a non-empty sequence of integers $i_1,i_2,\ldots,i_n\in\Al = \set{0,1,\ldots,\ell-1}$ such that
\[
	\alpha_{i_1}\alpha_{i_2}\cdots \alpha_{i_n} = \beta_{i_1}\beta_{i_2}\cdots \beta_{i_n}.
\]
It is well-known that the PCP is undecidable, even if $\Sigma=\Abin$ is the binary alphabet.
\end{definition}
\end{journal}

\begin{theorem}\label{thm:input:preserving}
For every fixed antimorphic permutation $\theta$ over $\Aks$ with $k\ge 2$ it is undecidable whether or not a given transducer is $\theta$-input-preserving.
\end{theorem}

\begin{journal}
\begin{proof}
Let $\alpha_0,\alpha_1,\ldots,\alpha_{\ell-1} \in \Sigma^+$ and $\beta_0,\beta_1,\ldots,\beta_{\ell-1} \in \Sigma^+$ be the PCP instance $\cA$.
We will define a transducer $\trt_\cA$ which accepts all pairs $(w,\theta(w))$ unless $w$ is a binary encoding of a word $uv$ where $u\in \Sigma^+$ and $v\in \Alp$ such that $v$ describes an integer sequence $i_1,i_2,\ldots,i_n$ that is a solution of \cA and $u$ is the corresponding solution word.
For the ease of notation, we assume that $\Sigma$ and \Al are two disjoints alphabet and we let $\Gamma = \Sigma\cup \Al$ be their union.
For $m = \ceil{\log_2\abs{\Gamma}}$, we let $h\colon \Gamma \to \Abinm$ be a morphic block code; \ie an encoding of $\Gamma$ into binary words of length $m$ such that $h(a) = h(b)$ implies $a = b$ for all $a,b\in \Gamma$.
Our goal is to define $\trt_\cA$ such that $\theta(w)\notin \trt_\cA(w)$ if and only if $w = h(uv)$ for $u\in\Sigma^+$, $v\in\Alp$, $n = \abs v$, and
\[
	u = \alpha_{v\pos n}\alpha_{v\pos{n-1}}\cdots \alpha_{v\pos 1} = \beta_{v\pos n}\beta{v\pos{n-1}}\cdots \beta_{v\pos 1}.
\]

The transducer $\trt_\cA$ will consist of $3$ effectively constructable components $\trt_R$, $\trt_\alpha$, and $\trt_\beta$.
Each component can be seen as a fully functional transducer such that $\trt_\cA$ becomes the union of the three transducers; this implies that
\[
	y\in \trt_\cA(x) \iff
	y \in \trt_{R}(x)\cup \trt_\alpha(x) \cup\trt_\beta(x).
\]
Each transducer component ``validates'' a certain property of a word $w$, by accepting all word pairs $(w,\theta(w))$ which do not have that property:
\begin{enumerate}[1.)]
\item
$\trt_{R}$ accepts $(w,\theta(w))$ if and only if $w \notin h(\Sigma^+\Alp)$;
\item
for $w\in h(uv)$ with $u\in\Sigma^+$ and $v\in\Alp$, $\trt_\alpha$ accepts $(w,\theta(w))$ if and only if $u \neq \alpha_{v\pos n}\alpha_{v\pos{n-1}}\cdots \alpha_{v\pos 1}$; and
\item
for $w\in h(uv)$ with $u\in\Sigma^+$ and $v\in\Alp$, $\trt_\beta$ accepts $(w,\theta(w))$ if and only if $u \neq \beta_{v\pos n}\beta_{v\pos{n-1}}\cdots \beta_{v\pos 1}$.
\end{enumerate}
The first component ensures that every pair $(w,\theta(w))$ that is not accepted by $\trt_\cA$ must have the desired form $w\in h(uv)$ with $u\in\Sigma^+$ and $v\in\Alp$.
Components $\trt_\alpha$ and $\trt_\beta$ ensure that
\[
	\alpha_{v\pos n}\alpha_{v\pos{n-1}}\cdots \alpha_{v\pos 1} = u =
	\beta_{v\pos n}\beta_{v\pos{n-1}}\cdots \beta_{v\pos 1}
\]
is the solution word that corresponds the integer sequence $v\pos n,v\pos{n-1},\ldots,v\pos 1$ if $(w,\theta(w))$ is not accepted by $\trt_\cA$.
Therefore, every word pair $(w,\theta(w))$ which is not accepted by $\trt_\cA$ yields a solution for \cA and, vice versa, every solution for \cA yields a word pair $(w,\theta(w))$ that cannot be accepted by $\trt_\cA$.
We conclude that $\trt_\cA$ is $\theta$-input-preserving if and only if the PCP instance \cA has no solution.
This implies that for fixed antimorphic $\theta$ over $\Aks$ with $k\ge2$ it is undecidable whether or not a given transducer is $\theta$-input-preserving because the PCP is undecidable.

Now, let us describe the transducer component $\trt_R$ and recall that it has to work over the alphabet $\Ak$.
It is well known that for any two regular languages $R_1$ and $R_2$ there effectively exists a transducer which accepts the relation $R_1\times R_2$.
There is $\trt_R$ such that $\trt_R = (\Aks\sm h(\Sigma^+\Alp))\times \Aks$.
It is easy to observe that we have $\trt_R(w) = \Aks$ if $w\notin h(\Sigma^+\Alp)$, and $\trt_R(w) = \es$ if $w\in h(\Sigma^+\Alp)$.
Therefore, we have $\theta(w)\notin\trt_R(w)$ if and only if $w\notin h(\Sigma^+\Alp)$.
Note that this in particular implies that, if $\theta(w)\notin \trt_R(w)$, then $w\in h(\Gamma^*)\sse (\Abinm)^*$.
The other two transducer components $\trt_\alpha$ and $\trt_\beta$ will only work over word pairs from $h(\Gamma^*)\times \theta(h(\Gamma^*))$.

\begin{figure}[ht!]
\begin{transducer}[node distance=6cm]
	\node [state,initial] (s) {$s_z$};
	\node [state,accepting,right of=s] (a) {$f_z$};
	\node [node distance=1cm,left=of s,anchor=east] {$\trt_z\colon{}$};
	\path (s) edge [loop above] node [above]
			{$\forall i\in\Al\colon{}(h(z_i),\theta(h(i)))$} ()
		(s) edge node [below]
			{$\forall i\in\Al,z'\in\Sigma^{\le\abs{z_i}}\sm\pref(z_i)\colon{}$\\
				$(h(z'),\theta(h(i)))$}
			node [above]
			{$\forall i,j\in\Al\colon(h(i),\theta(h(j)))$\\
				$\forall a,b\in\Sigma\colon(h(a),\theta(h(b)))$} (a)
		(a) edge [loop above] node [above]
			{$\forall a\in\Gamma\colon(h(a),\e)$\\
				$\forall a\in\Gamma\colon(\e,\theta(h(a)))$} ();
\end{transducer}
\caption{For $z\in\set{\alpha,\beta}$ the two transducers $\trt_{\alpha}$ and $\trt_\beta$ enforce that $w$ encodes a solution of the PCP instance $\cA$ if $\theta(w)\notin (\trt_{\alpha}+\trt_\beta)(w)$ and $w\in h(\Sigma^+\Alp)$.}
\label{trans:zpcp}
\end{figure}

Finally, we define the two transducers $\trt_\alpha$ and $\trt_\beta$ which are based on the words $\alpha_i$ and $\beta_i$, respectively.
For $z\in\set{\alpha,\beta}$ we define $\trt_z$ as shown in Fig.~\ref{trans:zpcp}.
For a pair of words $(x,y)\in \trt_z$, it is easy to see that $x\in h(\Gamma^*)$ and $y\in \theta( h(\Gamma^*))$.
Furthermore, the edges from the final state $f_z$ to itself ensure that if $(x,y)\in\theta$, then for all words $x'\in h(\Gamma^*)$ and $y'\in \theta(h(\Gamma^*))$, we have $(xx',yy')\in \trt_z$ (we will not leave the final state anymore once it is reached, unless the word pair is not defined over $h(\Gamma^*)\times \theta(h(\Gamma^*))$).
There are three possibilities to switch from state $s_z$ to the final state $f_z$:
\begin{enumerate}[1.)]
\item
we read a word from $h(\Al)$ in the first component and a words from $\theta(h(\Al))$ in the second component;
\item
we read a word from $h(\Sigma)$ in the first component and a words from $\theta(h(\Sigma))$ in the second component; or
\item
we read the word $\theta(h(i))$ with $i\in\Al$ in the second component and in the first component we read a word $h(z')$ such that $z'$ is not a prefix of $z_i$ and $z_i$ is not a prefix $z'$ because of the length restriction on $z'$.
\end{enumerate}
For $x\in h(\Gamma^*)$ let $u$ denote the longest word in $\Sigma^*$ such that $h(u)$ is a prefix of $x$ (thus, either $x = h(u)$ or $x = h(u i x')$ for an integer $i\in \Al$ and $x'\in\Gamma^*$);
and for $y\in\theta(h(\Gamma^*))$ let $v$ denote the longest word in $\Als$ such that $\theta(h(v))$ is a prefix of $y$ and let $n = \abs v$ (thus, either $y = \theta(h(v))$ or $y = \theta(h(y'av)) = \theta(h(v))\theta(h(a))\theta(h(y'))$ for a symbol $a\in \Sigma$ and $y'\in\Gamma^*$).
Because $\theta(h(v\pos n))\theta(h(v\pos{n-1}))\cdots \theta(h(v\pos 1))$ is a prefix of $y$ we obtain that the pair $(x,y)$ is accepted by $\trt_z$ if $u \neq z_{v\pos n}z_{v\pos{n-1}} \cdots z_{v\pos 1}$.
Conversely, if $u = z_{v\pos n}z_{v\pos{n-1}} \cdots z_{v\pos 1}$, then $(h(u),\theta(h(v)))$ labels a path from $s_z$ to $s_z$;
since there is no edge from $s_z$ which is labeled $(h(i),\e)$, $(\e,\theta(h(a)))$, or $(h(i),\theta(h(a)))$ for $i\in\Al$ and $a\in\Sigma$, we obtain that $(x,y)$ cannot not be accepted by $\trt_z$.

Suppose $\theta(w)\notin \trt_z(w)$ and $w\in h(uv)$ for words $u\in \Sigma^+$ and $v\in \Alp$.
Following our notion from the previous paragraph, $u$ is the longest word in $\Sigma^*$ such that $h(u)$ is a prefix of $w$, and $v$ is the longest word in $\Als$ such that $\theta(h(v))$ is a prefix of $\theta(w)$.
Therefore, we obtain that $u = z_{v\pos n} \cdots z_{v\pos 1}$.
\end{proof}
\end{journal}

This leads to the undecidability of the maximality problem of a regular language $L$ with respect to a $\theta$-transducer-property \Stt.

\begin{corollary}\label{cor:undecidable:maximality}
For every fixed antimorphic permutation $\theta$ over $\Aks$ with $k\ge 2$, it is undecidable whether or not the empty language $\es$ is maximal with respect to the property $\Stt$, for a given transducer $\trt$.
\end{corollary}

\begin{conference}
Note that a singleton language $\set{w}$ satisfies $\Stt$ if and only if $\theta(w)\notin \trt(w)$.
Thus, the corollary follows because $\es$ is maximal with property $\Stt$ if and only if $\trt$ is $\theta$-input-preserving.
\end{conference}

\begin{journal}
\begin{proof}
Clearly, the empty language satisfies $\Stt$.
For a word $w$, the language $\set{w}$ satisfies $\Stt$ if and only if $\theta(w)\notin \trt(w)$.
Therefore, $\es$ is maximal with property $\Stt$ if and only if $\trt$ is $\theta$-input-preserving.
Theorem~\ref{thm:input:preserving} concludes the proof.
\end{proof}
\end{journal}

\section{Undecidability of the \theta-PCP and the \theta-input-altering Transducer Problem}\label{sec:PCP}

Analogous to the undecidable PCP%
\begin{journal}
\ (see Definition~\ref{def:pcp})
\end{journal},
we introduce the \theta version of the PCP and prove that it is undecidable as well; see
Theorem~\ref{thm:tPCP}.
Further, we utilize the \theta version of the PCP in order to show that it is undecidable whether or not a transducer is \theta-input-altering; see Corollary~\ref{cor:input:altering}.

\begin{definition}\label{def:tPCP}
For a fixed antimorphic permutation $\theta$ over $\Aks$, we introduce the \emph{\theta-Post correspondence problem} (\theta-PCP): given words $\alpha_0,\alpha_1,\ldots,\alpha_{\ell-1}\in\Akp$ and $\beta_0,\beta_1,\ldots,\beta_{\ell-1}\in\Akp$, decide whether or not there exists a non-empty sequence of integers $i_1,\ldots,i_n\in\Al=\set{0,1,\ldots,\ell-1}$ such that
\[
	\alpha_{i_1}\alpha_{i_2}\cdots \alpha_{i_n} = \theta(\beta_{i_1}\beta_{i_2}\cdots \beta_{i_n}).
\]
\end{definition}

\begin{theorem}\label{thm:tPCP}
For every fixed antimorphic permutation $\theta$ over $\Aks$ with $k\ge 2$ the $\theta$-PCP is undecidable.
\end{theorem}

\begin{journal}
\begin{proof}
In order to prove that \theta-PCP is undecidable, we will state an effective reduction of any PCP instance \cA over alphabet $\Abin$ to a \theta-PCP instance \cT over alphabet \Ak such that \cA has a solution if and only if \cT has a solution.
Let $\alpha_0,\alpha_1,\ldots,\alpha_{\ell-1}\in\Abinp$ and $\beta_0,\beta_1,\ldots,\beta_{\ell-1}\in\Abinp$ be an instance of the PCP which we call $\cA$.

Note that $\theta$ and $\thetai$ are well-defined over $\Abin \sse \Ak$.
We define two morphisms $g,h$ on $\Abin^*$ such that
\begin{align*}
	g(0) &= 00, &
	g(1) &= 01, &
	h(0) &= 10, &
	h(1) &= 11.
\end{align*}
Note that for each pair of letters $z\in\Abin^2$ we have either $z\in h(\Abin)$ or $z\in g(\Abin)$.
Moreover, we let
\begin{align*}
	\gamma_j &= g(\alpha_j), & \delta_j &= \thetai(h(\beta_j^R)),
	&&\text{for }j = 0,\ldots,\ell-1, \\
	\gamma_{\ell} &= h(0), & \delta_{\ell} &= \thetai(g(0)), \\
	\gamma_{\ell+1} &= h(1), & \delta_{\ell+1} &= \thetai(g(1)).
\end{align*}
be the \theta-PCP instance \cT.

\begin{figure}[ht!]
\centering
\begin{tikzpicture}[font=\footnotesize,text depth=1ex]
	\draw [|-|] (0,0) -- (5.5,0);
	\node (A0) [anchor=south, align=center] at (.5,0) {$\gamma_{i_1}$};
	\node (A1) [node distance=1.5cm, right of=A0, align=center] {$\gamma_{i_2}$};
	\node (A2) [node distance=1.5cm,right of=A1,align=center] {$\cdots$};
	\node (A3) [node distance=1.5cm,right of=A2,align=center] {$\gamma_{i_n}$};

	\node (B0) [anchor=north, align=center] at (.5,0) {$\theta(\!\delta_{i'_1}\!)$};
	\node (B1) [node distance=.75cm,right of=B0,align=center] {$\theta(\!\delta_{i'_2}\!)$};
	\node (B2) [node distance=.75cm,right of=B1,align=center] {$\theta(\!\delta_{i'_3}\!)$};
	\node (B3) [node distance=.75cm,right of=B2,align=center] {$\theta(\!\delta_{i'_4}\!)$};
	\node (B4) [node distance=.75cm,right of=B3,align=center] {$\theta(\!\delta_{i'_5}\!)$};
	\node (B5) [node distance=.75cm,right of=B4,align=center] {$\cdots$};
	\node (B6) [node distance=.75cm,right of=B5,align=center] {$\theta(\!\delta_{i'_{m}}\!)$};
	
	\draw [decorate,decoration={brace,mirror}] (0,-.6) --
		node [below=3pt] {$g(w)$} (5.5,-.6);

	\draw [-|] (5.5,0) -- (11,0);
	\node (C0) [anchor=south, align=center] at (10.5,0) {$\gamma_{i'_1}$};
	\node (C1) [node distance=.75cm,left of=C0,align=center] {$\gamma_{i'_2}$};
	\node (C2) [node distance=.75cm,left of=C1,align=center] {$\gamma_{i'_3}$};
	\node (C3) [node distance=.75cm,left of=C2,align=center] {$\gamma_{i'_4}$};
	\node (C4) [node distance=.75cm,left of=C3,align=center] {$\gamma_{i'_5}$};
	\node (C5) [node distance=.75cm,left of=C4,align=center] {$\cdots$};
	\node (C6) [node distance=.75cm,left of=C5,align=center] {$\gamma_{i'_{m}}$};

	\node (D0) [anchor=north, align=center] at (10.5,0) {$\theta(\delta_{i_1})$};
	\node (D1) [node distance=1.5cm, left of=D0, align=center] {$\theta(\delta_{i_2})$};
	\node (D2) [node distance=1.5cm,left of=D1,align=center] {$\cdots$};
	\node (D3) [node distance=1.5cm,left of=D2,align=center] {$\theta(\delta_{i_n})$};

	\draw [decorate,decoration={brace}] (5.5,.6) --
		node [above] {$h(w^R)$} (11,.6);

\end{tikzpicture}
\caption{Transforming the solution $i_1,i_2,\ldots,i_n$ of the PCP instance \cA into the solution $i_1,i_2,\ldots,i_n$, $i'_{m},i'_{m-1},\ldots,i'_1$ of the $\theta$-PCP instance \cT; all variables are defined in the text.}
\label{fig:tpcp}
\end{figure}

First, let us show that if \cA has a solution than \cT has a solution as well.
Let $i_1,i_2,\ldots, i_n \in \Al$ with $n \ge 1$ be a solution of the PCP instance \cA and let $w$ be the word corresponding to this solution; \ie
\[
	w = \alpha_{i_1}\alpha_{i_2}\cdots \alpha_{i_n} = \beta_{i_1}\beta_{i_2}\cdots \beta_{i_n}.
\]
Figure~\ref{fig:tpcp} illustrates the following construction.
Let $m = \abs w$.
For $j = 1,\ldots,m$ we let $i'_j = \ell$ if $w\pos j= 0$ and $i'_j = \ell+1$ if $w\pos j = 1$; these indeces are chosen such that
\begin{align*}
	\gamma_{i'_{m}}\gamma_{i'_{m-1}}\cdots \gamma_{i'_{1}} &= h(w^R), \\
	\delta_{i'_{m}}\delta_{i'_{m-1}}\cdots \delta_{i'_{1}} &=
	\thetai(g(w\pos{m}))\thetai(g(w\pos{m-1})) \cdots \thetai(g(w\pos 1)) =
	\thetai(g(w)).
\end{align*}
The integer sequence $i_1,i_2,\ldots,i_n$, $i'_{m},i'_{m-1},\ldots,i'_1$ is a solution of the $\theta$-PCP instance $f(\alpha)$ because
\begin{alignat*}{2}
	\theta(\delta_{i_1}\cdots \delta_{i_n} \delta_{i'_{m}} \cdots \delta_{i'_{1}})
	&= \theta(\delta_{i'_{m}}\cdots \delta_{i'_{1}}) &&\cdot
		\theta( \delta_{i_{n}})\cdots \theta( \delta_{i_{1}}) \\
	&= \theta(\thetai(g(w))) &&\cdot
		\theta( \thetai(h(\beta_{i_{n}}^R)))
		 \cdots \theta( \thetai(h(\beta_{i_{1}}^R))) \\
	&= g(w) &&\cdot h(\beta_{i_{n}}^R) \cdots h(\beta_{i_{1}}^R) \\
	&= g(\alpha_{i_1})\cdots g(\alpha_{i_n}) &&\cdot h(w^R) \\
	&= \gamma_{i_1}\cdots \gamma_{i_n} &&\cdot \gamma_{i'_{m}} \cdots \gamma_{i'_{1}} .
\end{alignat*}

Vice versa, let $i_1,i_2,\ldots, i_n\in \A_{ell+2}$ with $n \ge 1$ be a solution of the $\theta$-PCP instance \cT and let $w$ be the word corresponding to this solution, that is,
\[
	w = \gamma_{i_1}\gamma_{i_2}\cdots \gamma_{i_n} = \theta(\beta_{i_1}\beta_{i_2}\cdots \beta_{i_n})
	= \theta(\beta_{i_n})\cdots \theta(\beta_{i_2})\theta(\beta_{i_1}).
\]
Recall that for every word $\gamma_{i_j}$ we have that either $\gamma_{i_j}\in g(\Abin^+)$ (in case $i_j < \ell$) or $\gamma_{i_j}\in h(\Abin)$ (in case $i_j \ge \ell$).
Since $g(\Abin)$ and $h(\Abin)$ contain mutually distinct two-letter words, for every pair of letters $p = w\pos[2r-1]{2r}$ with $r\in\N$: if $p \in g(\Abin)$, then $p$ is covered by a factor $\gamma_{i_j}$ with ${i_j} < \ell$; and if $p \in h(\Abin)$, then $p$ equals to a factor $\gamma_{i_j}$ with ${i_j} \ge \ell$.
Symmetrically, for $p = w\pos[2r-1]{2r}$ with $r\in\N$: if $p \in h(\Abin)$, then $p$ is covered by a factor $\theta(\delta_{i_j})$ with ${i_j} < \ell$; and if $p \in g(\Abin)$, then $p$ equals to a factor $\theta(\delta_{i_j})$ with ${i_j} > \ell$.

\begin{figure}[ht]
\centering
\begin{tikzpicture}[font=\footnotesize,text depth=1ex]
	\draw [|-|] (0,0) -- (5.5,0);
	\node (A0) [anchor=south, align=center] at (.5,0) {$\gamma_{i_1}$};
	\node (A1) [node distance=1.5cm, right of=A0, align=center] {$\gamma_{i_2}$};
	\node (A2) [node distance=1.5cm,right of=A1,align=center] {$\cdots$};
	\node (A3) [node distance=1.5cm,right of=A2,align=center] {$\gamma_{i_{n'}}$};

	\node (B0) [anchor=north, align=center] at (5,0) {$\theta(\!\delta_{i_j}\!)$};
	\node (B1) [node distance=.9cm,left of=B0,align=center] {$\theta(\!\delta_{i_{j\!{+}\!1}}\!)$};
	\node (B2) [node distance=.9cm,left of=B1,align=center] {$\theta(\!\delta_{i_{j\!{+}\!1}}\!)$};
	\node (B3) [node distance=.9cm,left of=B2,align=center] {$\theta(\!\delta_{i_{j\!{+}\!1}}\!)$};
	\node (B4) [node distance=.9cm,left of=B3,align=center] {$\cdots$};
	\node (B5) [node distance=.9cm,left of=B4,align=center] {$\theta(\!\delta_{i_n}\!)$};
	
	\draw [decorate,decoration={brace}] (0,.6) --
		node [above] {$\in g(\Abin^+)$} (5.5,.6);

	\draw [dashed] (5.5,0) -- (6.5,0);

	\draw [|-|] (6.5,0) -- (12,0);
	\node (C0) [anchor=south, align=center] at (7,0) {$\gamma_{i_j}$};
	\node (C1) [node distance=.9cm,right of=C0,align=center] {$\gamma_{i_{j\!{+}\!1}}$};
	\node (C2) [node distance=.9cm,right of=C1,align=center] {$\gamma_{i_{j\!{+}\!2}}$};
	\node (C3) [node distance=.9cm,right of=C2,align=center] {$\gamma_{i_{j\!{+}\!3}}$};
	\node (C4) [node distance=.9cm,right of=C3,align=center] {$\cdots$};
	\node (C5) [node distance=.9cm,right of=C4,align=center] {$\gamma_{i_{n}}$};

	\node (D0) [anchor=north, align=center] at (11.5,0) {$\theta(\delta_{i_1})$};
	\node (D1) [node distance=1.5cm, left of=D0, align=center] {$\theta(\delta_{i_2})$};
	\node (D2) [node distance=1.5cm,left of=D1,align=center] {$\cdots$};
	\node (D3) [node distance=1.5cm,left of=D2,align=center] {$\theta(\delta_{i_n'})$};

	\draw [decorate,decoration={brace}] (6.5,.6) --
		node [above] {$\in h(\Abin^+)$} (12,.6);

\end{tikzpicture}
\caption{Transforming the solution $i_1,i_2,\ldots,i_n$ of the $\theta$-PCP instance \cT into the solution $i_1,i_2,\ldots,i_{n'}$ of the PCP instance \cA; all variables are defined in the text.}
\label{fig:tpcp2}
\end{figure}

Consider the case where $i_1 < \ell$.
Figure~\ref{fig:tpcp2} illustrates the following construction.
In this case, $\gamma_{i_1} = g(\alpha_{i_1})$ is a prefix of $w$ and $\theta(\delta_{i_1}) = h(\beta_{i_1}^R)$ is a suffix of $w$; thus, $w\pos[1]{2}\in g(\Abin)$ and $w\pos[\abs w-1]{\abs w} \in h(\Abin)$.
Further, we obtain that $i_n \ge \ell$ because $\gamma_{i_n}$ has to cover $w\pos[\abs w-1]{\abs w} \in h(\Abin)$.
There exists an integer $n'$ with $1\le n' < n$ such that $i_1,i_2,\ldots, i_{n'}<\ell$ but $i_{n'+1} \ge \ell$.
We will show that the sequence $i_1,i_2,\ldots,i_{n'}$ is a solution of the PCP instance \cA by comparing the longest prefix of $w$ which belongs to $g(\Abin^+)$ with the longest suffix of $w$ which belongs to $h(\Abin^+)$.
Let $m$ be an even integer such that $w\pos[1]{m} \in g(\Abin^+)$ but $w\pos[m+1]{m+2} \in h(\Abin)$.
Because $i_{n'+1}$ has to match with the first letter pair in $w$ which belongs to $h(\Abin)$, it is not difficult to see that
\[
	w\pos[1]{m} = \gamma_{i_1}\gamma_{i_2}\cdots \gamma_{i_{n'}} = g(\alpha_{i_1}\alpha_{i_2}\cdots \alpha_{i_{n'}}).
\]
Because $w\pos[1]{m} \in g(\Abin^+)$ and $w\pos[m+1]{m+2} \in h(\Abin)$, there exists an integer $j < n$ such that $i_j,i_{j+1}\ldots,i_n \ge \ell$, $i_{j-1} < \ell$, and
\begin{align*}
	w\pos[1]{m} &= \theta(\delta_{i_j}\delta_{i_{j+1}}\cdots \delta_{i_n})
		= \theta(\delta_{i_n})\cdots\theta(\delta_{i_{j+1}}) \theta(\delta{i_j}).
\end{align*}
Due to the design of the word pairs $(\gamma_{\ell}, \delta_{\ell})$ and $(\gamma_{\ell+1}, \delta_{\ell+1})$ and because
\[
\theta(\delta_{i_n})\cdots \theta(\delta_{i_{j+1}})\theta(\delta_{i_j}) = g(\alpha_{i_1}\alpha_{i_2}\cdots \alpha_{i_{n'}})
\]
is a prefix of $w$, we have that $\gamma_{i_j}\gamma_{i_{j+1}}\cdots \gamma_{i_n} = h((\alpha_{i_1}\alpha_{i_2}\cdots \alpha_{i_{n'}})^R)$ is a suffix of $w$.
Since $i_{j-1} < \ell$, we see that this suffix $h((\alpha_{i_1}\alpha_{i_2}\cdots \alpha_{i_{n'}})^R)$ of $w$ is preceded by a letter pair from $g(\Abin)$.
This implies that the suffix $\theta(\delta_{i_{n'}}) \cdots \theta(\delta_{i_2})\theta(\delta_{i_1})$ of $w$ equals $h((\alpha_{i_1}\alpha_{i_2}\cdots \alpha_{i_{n'}})^R)$.
Therefore,
\begin{align*}
	h((\alpha_{i_1}\alpha_{i_2}\cdots \alpha_{i_{n'}})^R)
	&= \theta(\delta_{i_{n'}}) \cdots \theta(\delta_{i_2})\theta(\delta_{i_1}) \\
	&= h(\beta_{i_{n'}}^R) \cdots h(\beta_{i_2}^R)h(\beta_{i_1}^R) \\
	&= h((\beta_{i_{1}}\beta_{i_{2}} \cdots \beta_{i_{n'}})^R).
\end{align*}
We conclude that $\alpha_{i_1}\alpha_{i_2}\cdots \alpha_{i_{n'}} = \beta_{i_{1}}\beta_{i_{2}} \cdots \beta_{i_{n'}}$ and, therefore, $i_1,i_2,\ldots,i_{n'}$ is a solution of the PCP instance \cA.

The case when $i_1 \ge \ell$ can be treated analogously, where we compare the longest prefix of $w$ which belongs to $h(\Abin^+)$ and the longest suffix of $w$ which belongs to $g(\Abin^+)$.
In this case, there exists $n' \le n$ such that $i_{n'},i_{n'+1},\ldots, i_n$ is a solution of the PCP instance \cA.
\end{proof}
\end{journal}

We can utilize the $\theta$-PCP in order to prove that it is undecidable whether or not a transducer is $\theta$-input-altering, even for one-state transducers.

\begin{corollary}\label{cor:input:altering}
For every fixed antimorphic permutation $\theta$ over $\Aks$ with $k\ge 2$ it is undecidable whether or not a given (one-state) transducer is $\theta$-input-altering.
\end{corollary}

\begin{conference}
Corollary~\ref{cor:input:altering} follows because the $\theta$-PCP instance $\alpha_0,\ldots,\alpha_{\ell-1}$, $\beta_0,\ldots,\beta_{\ell-1}$ has a solution if and only if the the one-state transducer $\trt$ with edges $q\xra{(\alpha_i,\theta^2(\beta_i))}q$ for $i = 0,\ldots,\ell-1$ is not $\theta$-input-altering.
\end{conference}

\begin{journal}
\begin{proof}
Let $\alpha_0,\alpha_1,\ldots,\alpha_{\ell-1}\in\Akp$ and $\beta_0,\beta_1,\ldots,\beta_{\ell-1}\in\Akp$ be the $\theta$-PCP instance $\cA$.
We let $\trt_\cA$ be the one-state transducer shown in Fig.~\ref{trans:input-altering}.
Clearly, we have $y \in \trt_\cA(x)$ if and only if there exists an integer sequence $i_1,i_2,\ldots,i_n\in \Al$ such that $x = \alpha_{i_1}\alpha_{i_2}\cdots \alpha_{i_n}$ and $y = \theta^2(\beta_{i_1})\theta^2(\beta_{i_2})\cdots \theta^2(\beta_{i_n}) = \theta^2(\beta_{i_1}\beta_{i_2}\cdots \beta_{i_n})$; note that $\theta^2$ is always morphic, even if $\theta$ is not.

\begin{figure}[ht]
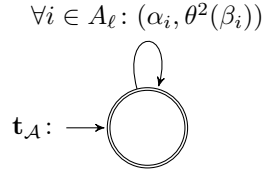

\begin{transducer}
	\node [state,initial,accepting] (q) {};
	\node [node distance=1cm,left=of q,anchor=east] {$\trt_\cA\colon{}$};
	\path [->] (q) edge [loop above] node [above]
		{$\forall i\in \Al\colon(\alpha_i,\theta^2(\beta_i))$} ();
\end{transducer}
\caption{$\trt_\cA$ encodes the $\theta$-PCP instance $\alpha_0,\alpha_1,\ldots,\alpha_{\ell-1}$, $\beta_0,\beta_1,\ldots,\beta_{\ell-1}$.}
\label{trans:input-altering}
\end{figure}

Recall that it is allowed for \theta-input-altering transducers to accept the empty word pair $(\e,\e)$.
We have $w \in \thetai(\trt_\cA(w))$ for some word $w\in \Akp$ if and only if there exists an integer sequence $i_1,i_2,\ldots,i_n$ such that
\[
	\alpha_{i_1}\alpha_{i_2}\cdots\alpha_{i_n} = w
	= \thetai\left(\theta^2(\beta_{i_1}\beta_{i_2}\cdots \beta_{i_n})\right)
	= \theta(\beta_{i_1}\beta_{i_2}\cdots \beta_{i_n}).
\]
Therefore, $\trt_\cA$ is $\theta$-input-altering if and only if the $\theta$-PCP instance \cA has a solution.
Theorem~\ref{thm:tPCP} concludes the proof.
\end{proof}
\end{journal}

\setcounter{propcounter}{0}

\begin{journal}
\section{A Hierarchy of DNA-related \theta-transducer Properties}
\label{sec:dna-properties}

In \cite{kari2002codes,HuKaKo:2003,KKLW:2003} the authors consider numerous properties of languages inspired by reliability issues in DNA computing.
Let $\theta$ be defined over $\As$ and assume that $\theta^2 = \id$ since in the DNA setting $\theta=\dna$ is an involution.
The relationships between some of the defined $3$-independent DNA-related properties are displayed in Fig.~\ref{fig:dna:properties}.
All properties have in common that they forbid certain ``constellations'' of words.
Consider a language $L\sse \Ap$ and two words $uwv,\theta(xwy)\in\Ap$ with $w \neq \e$ as shown in the top property in Fig.~\ref{fig:dna:properties}.
The same notation can be employed for all properties in the figure, where some properties require that $x$, $y$, $u$, or $v$ are empty, \eg for $x=y=\e$ we obtain the $\theta$-compliant property.
In the case of \theta-nonoverlapping all of $x,y,u,v$ are empty and
\begin{dnaprop}
\label{Prop:nol}
a language $L$ is \emph{$\theta$-nonoverlapping} if for all $w\in \Ap$, we have $w\notin L$ or $\theta(w)\notin L$.
This is equivalent to require that $L\cap\theta(L) = \es$.
\end{dnaprop}

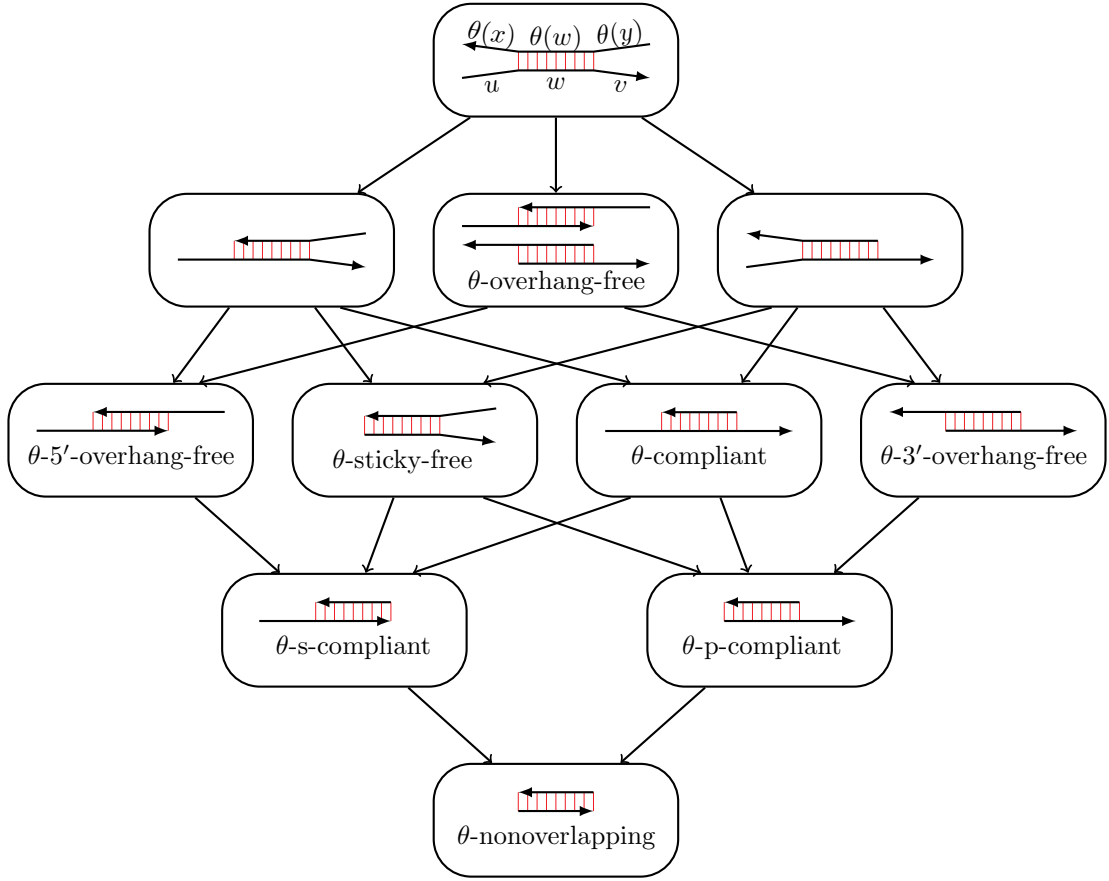
\begin{figure}[t!]
\centering
\begin{tikzpicture}[font=\footnotesize,node distance=1cm and .5cm,
		outer/.style={draw,thick,inner sep=2pt,rounded corners=.5cm,
			minimum width=3.25cm,minimum height=1.5cm},
		inner/.style={inner sep=0pt,anchor=center,
			minimum width=0pt,minimum height=0pt},
		bond/.style={Red,ultra thin}]
	\def\HALF{3.75cm/2}
	\newcommand{\BOND}{
		\foreach \x in {-.5,-.375,...,.5}
			\draw [bond] (\x,0) -- +(0,-.25);}
	\node (kcode) [outer]
	{\begin{tikzpicture}[rounded corners=0pt,inner,sloped,
			text height=1.5ex,text depth=.5ex]
		\BOND
		\draw [latex-,thick] (-1.25,.1) -- node [above] {$\theta(x)$} (-.5,0)
			-- node [above] {$\theta(w)$} (.5,0) -- node [above] {$\theta(y)$} (1.25,.1);
		\draw [-latex,thick] (-1.25,-.35) -- node [below] {$u$} (-.5,-.25)
			-- node [below] {$w$} (.5,-.25) -- node [below] {$v$}(1.25,-.35);
	\end{tikzpicture}};
	\node (OF) [outer,below=of kcode]
	{\begin{tikzpicture}[rounded corners=0pt]
		\BOND
		\draw [latex-,thick] (-1.25,0) -- (.5,0);
		\draw [-latex,thick] (-.5,-.25) -- (1.25,-.25);
		\begin{scope}[yshift=.5cm]
			\BOND
			\draw [latex-,thick] (-.5,0) -- (1.25,0);
			\draw [-latex,thick] (-1.25,-.25) -- (.5,-.25);
		\end{scope}
		\node [inner] at (0,-.5) {$\theta$-overhang-free};
	\end{tikzpicture}};
	\node (OFL) [outer,left=of OF]
	{\begin{tikzpicture}[rounded corners=0pt]
		\BOND
		\draw [latex-,thick] (-.5,0) -- (.5,0) -- (1.25,.1);
		\draw [-latex,thick] (-1.25,-.25) -- (.5,-.25) -- (1.25,-.35);
	\end{tikzpicture}};
	\node (OFR) [outer,right=of OF]
	{\begin{tikzpicture}[rounded corners=0pt]
		\BOND
		\draw [latex-,thick] (-1.25,.1) -- (-.5,0) -- (.5,0);
		\draw [-latex,thick] (-1.25,-.35) -- (-.5,-.25) -- (1.25,-.25);
	\end{tikzpicture}};
	\node (5OF) [outer,below=of OFL,xshift=-\HALF]
	{\begin{tikzpicture}[rounded corners=0pt]
		\BOND
		\draw [latex-,thick] (-.5,0) -- (.5,0) -- (1.25,0);
		\draw [-latex,thick] (-1.25,-.25) -- (-.5,-.25) -- (.5,-.25);
		\node [inner] at (0,-.6) {$\theta$-$5'$-overhang-free};
	\end{tikzpicture}};
	\node (SF) [outer,right=of 5OF]
	{\begin{tikzpicture}[rounded corners=0pt,anchor=center]
		\BOND
		\draw [latex-,thick] (-.5,0) -- (.5,0) -- (1.25,.1);
		\draw [-latex,thick] (-.5,-.25) -- (.5,-.25) -- (1.25,-.35);
		\node [inner] at (0,-.6) {$\theta$-sticky-free};
	\end{tikzpicture}};
	\node (comp) [outer,right=of SF]
	{\begin{tikzpicture}[rounded corners=0pt,anchor=center]
		\BOND
		\draw [latex-,thick] (-.5,0) -- (.5,0);
		\draw [-latex,thick] (-1.25,-.25) -- (1.25,-.25);
		\node [inner] at (0,-.6) {$\theta$-compliant};
	\end{tikzpicture}};
	\node (3OF) [outer,right=of comp]
	{\begin{tikzpicture}[rounded corners=0pt,anchor=center]
		\BOND
		\draw [latex-,thick] (-1.25,0) -- (.5,0);
		\draw [-latex,thick] (-.5,-.25) -- (1.25,-.25);
		\node [inner] at (0,-.6) {$\theta$-$3'$-overhang-free};
	\end{tikzpicture}};
	\node (SC) [outer,below=of SF,xshift=-\HALF/2]
	{\begin{tikzpicture}[rounded corners=0pt]
		\BOND
		\draw [latex-,thick] (-.5,0) -- (.5,0);
		\draw [-latex,thick] (-1.25,-.25) -- (.5,-.25);
		\node [inner] at (0,-.6) {$\theta$-s-compliant};
	\end{tikzpicture}};
	\node (PC) [outer,below=of comp,xshift=\HALF/2]
	{\begin{tikzpicture}[rounded corners=0pt]
		\BOND
		\draw [latex-,thick] (-.5,0) -- (.5,0);
		\draw [-latex,thick] (-.5,-.25) -- (1.25,-.25);
		\node [inner] at (0,-.6) {$\theta$-p-compliant};
	\end{tikzpicture}};
	\node (NO) [outer,below=of SC,xshift=\HALF*1.5]
	{\begin{tikzpicture}[rounded corners=0pt]
		\BOND
		\draw [latex-,thick] (-.5,0) -- (.5,0);
		\draw [-latex,thick] (-.5,-.25) -- (.5,-.25);
		\node [inner] at (0,-.6) {$\theta$-nonoverlapping};
	\end{tikzpicture}};

	\path [thick,->]
		(kcode) edge (OF) edge (OFL) edge (OFR)
		(OF.-40) edge (3OF.140)
		(OF.-140) edge (5OF.40)
		(OFL) edge (5OF) edge (SF)
		(OFL.-40) edge (comp.140)
		(OFR) edge (3OF) edge (comp)
		(OFR.-140) edge (SF.40)
		(5OF) edge (SC)
		(SF) edge (SC)
		(SF.-40) edge (PC.140)
		(comp) edge (PC)
		(comp.-140) edge (SC.40)
		(3OF) edge (PC)
		(SC) edge (NO)
		(PC) edge (NO);

\end{tikzpicture}
\caption{Correlation of various $3$-independent DNA language properties.
For each property the forbidden constellation of words (or single strands) is depicted.
Words are represented as arrows such that the first letter (the $5'$-end) is the blunt end of the arrow and last letter (the $3'$-end) is the arrow tip.
Red, vertical lines represent bonding between \theta-complementary parts of the two words.}
\label{fig:dna:properties}
\end{figure}

For all properties, except \theta-nonoverlapping, the language $L$ has property $P$, if  $uwv\in L$ and $\theta(xw y)\in L$ implies that $uvxy=\e$.
For example,
\begin{dnaprop}
\label{Prop:compliant}
a language $L$ is \emph{$\theta$-compliant} if for all
$w\in\Ap$ and $u,v\in\As$, we have $uwv,\theta(w)\in L \implies uv=\e$; and
\end{dnaprop}
\begin{dnaprop}
\label{Prop:5OF}
a language $L$ is \emph{$\theta$-$5'$-overhang-free} if for all
$w\in\Ap$ and $u,y\in\As$, we have $uw,\theta(wy)\in L \implies uy=\e$.
\end{dnaprop}

Previous papers considered the \emph{strict version} only for some of the properties.
Here, we generalize the concept of strict properties such that if $uwv\in L$ and $\theta(xwy)\in L$, then $L$ does not satisfy the strict property $P^\strict$ (even if $uvxy = \e$).
For example,
\begin{dnaprop}
\label{Prop:Scompliant}
a language $L$ is \emph{strictly $\theta$-compliant} if for all
$w\in\Ap$ and for all $u,v\in\As$, we have $uwv\notin L$ or $\theta(w)\notin L$; and
\end{dnaprop}
\begin{dnaprop}
\label{Prop:S5OF}
a language $L$ is \emph{strictly $\theta$-$5'$-overhang-free} if for all
$w\in\Ap$ and $u,y\in\As$, we have $uw\notin L$ or $\theta(wy)\notin L$.
\end{dnaprop}
Note that $\theta$-nonoverlapping is actually a strict property while its ``normal version'' would be the property that is trivially satisfied by every language in \Ap.

Furthermore, we introduce the \emph{weak version} of a property which follows the concept of classic code properties like the (weakly) overlap-free property where it is allowed for a word to overlap with itself, but not with another word:
for a language $L$ which satisfied the weak property $P^\weak$, if the words $uwv$ and $\theta(xwy)$ belong to $L$, then $uvxy=\e$ or $uwv = \theta(xwy)$.
For example,
\begin{dnaprop}
\label{Prop:W5OF}
a language $L$ is \emph{weakly $\theta$-$5'$-overhang-free} if for all
$w\in\Ap$ and $u,y\in\As$, we have $uw,\theta(wy)\in L$ implies $uy=\e$ or $uw=\theta(wy)$.
\end{dnaprop}
Note that for some properties, like $\theta$-compliant, the weak property $P^\weak$ coincides with the (normal) property $P$.

If a language $L$ satisfies the strict property $P^\strict$, then it also satisfies the corresponding (normal) property $P$; and if $L$ satisfies the (normal) property $P$, then it also satisfies the corresponding weak property $P^\weak$.
Furthermore, there is a normal, strict, and weak hierarchy of properties which is shown in Fig.~\ref{fig:dna:properties}, where \theta-nonoverlapping only exists in the strict hierarchy.
For all three hierarchies an arrow $P^\wildx\to Q^\wildx$ (for $\wildx\in\set{\e,\strict,\weak}$) between two properties $P^\wildx$ and $Q^\wildx$ means that if a language $L$ satisfies property $P^\wildx$, then it also satisfies property $Q^\wildx$.

Let us discuss how these properties can be described as \theta-transducer properties.
The type of the property (\WP or \SP) and the type of the transducer (unrestricted, $\theta$-input-altering, $\theta$-input-preserving) is important when it comes to the complexity of the satisfaction problem and the decidability of the maximality problem; see Table~\ref{tab:un-decidability}.
Firstly, observe that $L$ is $\theta$-nonoverlapping if $L$ satisfies the \theta-transducer property \Stt[\trt_\id] where $\trt_\id$ is a transducer realizing the identity relation.
Since any strict property, including $\theta$-nonoverlapping, is not satisfied by a singleton language $\set{w}$ that consists of one \theta-palindrome $w = \theta(w)$, strict properties cannot be described as \SPs by a \theta-input-altering transducer or as \WPs, according to Remark~\ref{rem:one-word}.

Figure~\ref{trans:dna} shows two families of transducers which are capable of describing any of the DNA-related properties that we introduced in this section.
Depending on whether or not $u$ (\resp $v,x,y$) is empty one has to omit a set of edges in each transducer.
The \SPs $\Stt[\trt_{\strict}]$ describe the strict properties, the \SPs $\Stt[\trt_{\weak}]$ describe the normal properties, and the \WPs $\Wtt[\trt_{\weak}]$ describe the weak properties.
If we omit red and orange edges (\ie $xy=\e$), then $\trt_{\weak}$ is \theta-input-altering because the input word is strictly longer than the output word.
Therefore, $\Stt[\trt_{\weak}] = \Wtt[\trt_{\weak}]$, \ie the normal property coincides with the corresponding weak property.
The case when all blue and green edges are omitted is symmetric when input and output swap roles.
We demonstrate this construction in Examples~\ref{ex:compliant} and~\ref{ex:5OF}.

\begin{figure}[ht!]
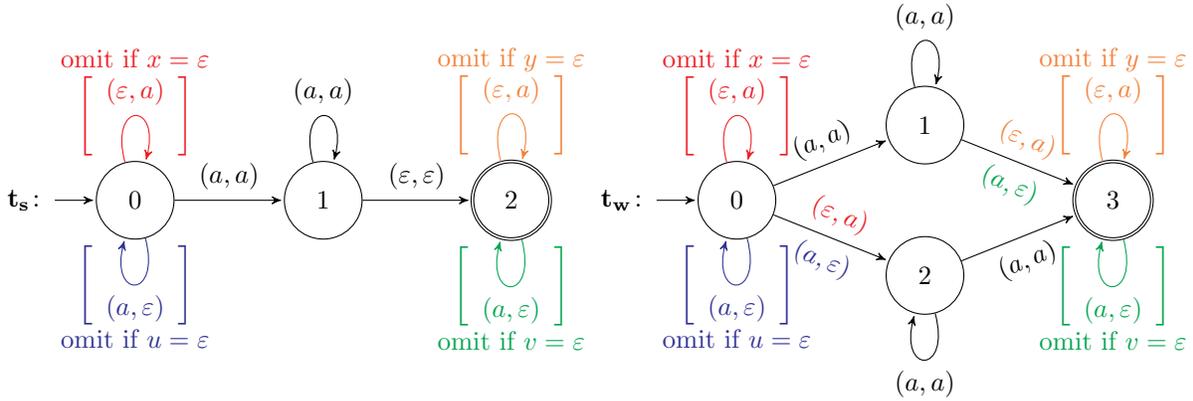

\begin{transducer}[node distance=1cm and 2.5cm]
	\node [state,initial] (q0) {$0$};
	\node [node distance=1cm,left=of q0,anchor=east] {$\trt_{\strict}\colon$};
	\node [state,right=of q0] (q1) {$1$};
	\node [state,accepting,right=of q1] (q2) {$2$};

	\path
		(q0) edge [loop below,Blue] node [below] {$(a,\e)$}
			node [anchor=center] {$\Biggl[\phantom{MMM}\Biggr]$}
			node [below=2.5ex] {omit if $u=\e$} ()	
		(q0) edge [loop above,Red] node [above] {$(\e,a)$}
			node [anchor=center] {$\Biggl[\phantom{MMM}\Biggr]$}
			node [above=2.5ex,text depth=.25ex] {omit if $x=\e$} ()	
		(q0) edge node [above] {$(a,a)$} (q1)
		(q1) edge [loop above] node [above] {$(a,a)$} ()	
		(q1) edge node [above] {$(\e,\e)$} (q2)
		(q2) edge [loop below,Green] node [below] {$(a,\e)$}
			node [anchor=center] {$\Biggl[\phantom{MMM}\Biggr]$}
			node [below=2.5ex] {omit if $v=\e$} ()	
		(q2) edge [loop above,Orange] node [above] {$(\e,a)$}
			node [anchor=center] {$\Biggl[\phantom{MMM}\Biggr]$}
			node [above=2.5ex,text depth=.25ex] {omit if $y=\e$} ();

	\begin{scope}[xshift=8cm]
	\node [state,initial] (p0) {$0$};
	\node [node distance=1cm,left=of p0,anchor=east] {$\trt_{\weak}\colon$};
	\node [state,above right=of p0] (p1) {$1$};
	\node [state,below right=of p0] (p2) {$2$};
	\node [state,accepting,below right=of p1] (p3) {$3$};

	\path [sloped] (p0) edge [loop below,Blue] node [below] {$(a,\e)$}
			node [anchor=center] {$\Biggl[\phantom{MMM}\Biggr]$}
			node [below=2.5ex] {omit if $u=\e$} ()	
		(p0) edge [loop above,Red] node [above] {$(\e,a)$}
			node [anchor=center] {$\Biggl[\phantom{MMM}\Biggr]$}
			node [above=2.5ex,text depth=.25ex] {omit if $x=\e$} ()	
		(p0) edge node [above] {$(a,a)$} (p1)
		(p1) edge [loop above] node [above] {$(a,a)$} ()	
		(p1) edge node [above,Orange] {$(\e,a)$}
			node [below,Green] {$(a,\e)$} (p3)
		(p0) edge node [above,Red] {$(\e,a)$}
			node [below,Blue] {$(a,\e)$} (p2)
		(p2) edge [loop below] node [below] {$(a,a)$} ()	
		(p2) edge node [below] {$(a,a)$} (p3)
		(p3) edge [loop below,Green] node [below] {$(a,\e)$}
			node [anchor=center] {$\Biggl[\phantom{MMM}\Biggr]$}
			node [below=2.5ex] {omit if $v=\e$} ()	
		(p3) edge [loop above,Orange] node [above] {$(\e,a)$}
			node [anchor=center] {$\Biggl[\phantom{MMM}\Biggr]$}
			node [above=2.5ex,text depth=.25ex] {omit if $y=\e$} ();
	\end{scope}
\end{transducer}
\caption{The family of transducers which describes all properties shown in Fig.~\ref{fig:dna:properties}.
Each of the two transducer families describes 16 different transducers: We can either omit or include each of the red, orange, blue and green edges.
These edges are omitted depending on the property that is described, for example, omit all red edges if $x=\e$ in Fig.~\ref{fig:dna:properties}.}
\label{trans:dna}
\end{figure}

\begin{example}\label{ex:compliant}
Let $\trt_\strict^C$ and $\trt_\weak^C$ be the two transducers that are obtained by omitting all red and orange edges in $\trt_\strict$ and $\trt_\weak$ (Fig.~\ref{trans:dna}), respectively.
Then $\Stt[\trt_{\strict}^C]$ is the strict \theta-compliant property, whereas $\Stt[\trt_{\weak}^C]$ is the (normal) $\theta$-compliant property.
Since $\trt_\weak^C$ is $\theta$-input-altering, $\Stt[\trt_{\weak}^C]$ is equal to $\Wtt[\trt_\weak^C]$ and the properties \theta-compliant and weak \theta-compliant coincide.
\end{example}

\begin{example}\label{ex:5OF}
Let $\trt_\strict^{5OF}$ and $\trt_\weak^{5OF}$ be the two transducers that are obtained by omitting all red and green edges in $\trt_\strict$ and $\trt_\weak$ (Fig.~\ref{trans:dna}), respectively.
Then $\Stt[\trt_{\strict}^{5OF}]$ is the strict \theta-$5'$-overhang-free property, $\Stt[\trt_{\weak}^{5OF}]$ is the (normal) \theta-$5'$-overhang-free property, and $\Wtt[\trt_{\weak}^{5OF}]$ is the weak \theta-$5'$-overhang-free property.

Observe that the word $z = \bA\bA\bC\bG$ can have a \theta-$5'$-overhang with itself (as $x = \bA\bA$, $w = \theta(w)= \bC\bG$, and $y = \bT\bT$).
As expected, $\trt_\weak^{5OF}$ does accept the word pair $(\bA\bA\bC\bG,\bC\bG\bT\bT)$ and, therefore, the singleton language $\set{z}$ does not satisfy the (normal) \theta-$5'$-overhang-free property $\Stt[\trt_{\weak}^{5OF}]$, however, $\set{z}$ does satisfy the weak \theta-$5'$-overhang-free property $\Wtt[\trt_{\weak}^{5OF}]$.
\end{example}

Lastly, note that the (strict, weak) \theta-overhang-free property is different from the other properties in Fig.~\ref{fig:dna:properties} in so far that it forbids two word constellations: \theta-$5'$-overhangs and \theta-$3'$-overhangs.
This property can be described by a transducer which contains two components, where one component covers the \theta-$5'$-overhangs and the other component covers the \theta-$3'$-overhangs.

\end{journal}

\section{Conclusions}
We have defined a transducer-based method for describing DNA code properties which is
strictly more expressive than the trajectory method. In doing so, the satisfaction question remains
efficiently decidable. The maximality question for some types of properties is decidable, but
it is undecidable for others. While some versions of the maximality question for
trajectory  properties are decidable, the case of any given pair of regular trajectories and
any given regular language is not addressed in~\cite{Dom:2007}, so we consider this to be
an interesting problem to solve.

The maximality questions are phrased in terms of any fixed antimorphic permutation. This
direction of generalizing decision questions is also applied to the classic Post Correspondence
Problem, where we demonstrate that it remains undecidable. A consequence of this is that
the question of whether a given transducer is $\theta$-input-altering is also
undecidable. It is interesting to note that if, instead of fixing $\theta$, we fix the transducer $\trt$
to be the identity, or the transducer defining the \SP $\cH$
(see Fig.~\ref{trans:pH} in Sect.~\ref{sec:expressiveness}), then the question
of whether or not
\jcmath{\theta(L)\cap \trt(L) = \es}
is decidable (given any regular language $L$ and antimorphic permutation $\theta$).

The topic of studying description methods for code properties requires further attention. One
important aim is the actual implementation of the algorithms, as it is already done for
several classic code properties \cite{FAdo,Laser}. An immediate plan is to incorporate in those implementations what we
know about DNA code properties. Another aim is to increase the expressive power of our
description methods.
The formal method of~\cite{Jurg:1999} is quite expressive,
using a certain type
of first order formulae to describe properties. It could perhaps be further
worked out in a way that some of these formulae can be mapped to transducers.
We also note that if the defining method is too expressive then
even the satisfaction problem could become undecidable; see for example the method of
multiple sets of trajectories in \cite{DomSal:2006}.


\end{document}